\documentclass[a4paper]{article}
\usepackage[margin=30mm]{geometry}

\newif\ifprint
\newif\ifjamemo
\newif\ifreview

\usepackage{hamada}
\usepackage{hamada-time}

\usepackage{ascmac}

\usepackage{amsmath}
\usepackage{cleveref}

\usepackage[ruled]{algorithm}
\usepackage[noend]{algorithmic}

\crefname{lemma}{Lemma}{Lemmas}
\crefname{theorem}{Theorem}{Theorems}
\crefname{proposition}{Proposition}{Propositions}
\crefname{definition}{Definition}{Definitions}
\crefname{example}{Example}{Examples}
\crefname{corollary}{Corollary}{Corollaries}
\crefname{equation}{Equation}{Equations}
\crefname{subsubsection}{Section}{Sections}
\crefname{subsection}{Section}{Sections}
\crefname{section}{Section}{Sections}

\crefname{algorithm}{Algorithm}{Algorithms}
\crefname{ALC@unique}{Line}{Lines}
\crefname{figure}{Figure}{Figures}
\crefname{table}{Table}{Tables}

\newenvironment{calgorithmic}[1][1]{\begin{algorithmic}[#1]\setcounter{ALC@unique}{0}}{\end{algorithmic}}

\usepackage{amsfonts}

\newcommand{\pcase}[2]{\left\{\begin{array}{#1}#2\end{array}\right.} %

\jamemofalse
\newtheorem{definition}{Definition}
\newtheorem{theorem}{Theorem}
\newtheorem{lemma}{Lemma}

\newtheorem{example}{Example}
\newtheorem{corollary}{Corollary}

\def\squarebox#1{\hbox to #1{\hfill\vbox to #1{\vfill}}}
\def\qed{\hspace*{\fill}%
        \vbox{\hrule\hbox{\vrule\squarebox{.667em}\vrule}\hrule}\smallskip}
\newenvironment{proof}{\begin{trivlist}%
\item[\hspace{\labelsep}{\em\noindent Proof.~}]}{\qed\end{trivlist}}

\crefname{appendix}{Appendix}{Appendices}
\title{Refined Computational Complexities of Hospitals/Residents Problem with Regional Caps \footnote{Supported by JSPS KAKENHI Grant Number JP20K11677}}
\date{July 7, 2021}
\reviewtrue

\jamemofalse

\author{
Koki Hamada
\footnote{NTT Corporation, 3-9-11, Midori-cho, Musashino-shi, Tokyo 180-8585, Japan,
E-mail: koki.hamada.rb@hco.ntt.co.jp}
\footnote{Graduate School of Informatics, Kyoto University, Yoshida-Honmachi, Sakyo-ku Kyoto 606-8501, Japan}
\and Shuichi Miyazaki
\footnote{Academic Center for Computing and Media Studies, Kyoto University, Yoshida-Honmachi, Sakyo-ku, Kyoto 606-8501, Japan,
E-mail: shuichi@media.kyoto-u.ac.jp}
}

\begin{document}
\maketitle

\begin{abstract}
The Hospitals/Residents problem (HR) is a many-to-one matching problem whose solution concept is stability. 
It is widely used in assignment systems such as assigning medical students (residents) to hospitals.
To resolve imbalance in the number of residents assigned to hospitals, an extension called HR with regional caps (HRRC) was introduced.
In this problem, a positive integer (called a regional cap) is associated with a subset of hospitals (called a region),
and the total number of residents assigned to hospitals in a region must be at most its regional cap. 
Kamada and Kojima \cite{KK10,KK15} defined strong stability for HRRC and demonstrated that a strongly stable matching does not necessarily exist.
Recently, Aziz et al.~\cite{ABB20} proved that the problem of determining if a strongly stable matching exists is NP-complete in general.
In this paper, we refine Aziz et al.'s result by investigating the computational complexity of the problem
in terms of the length of preference lists, the size of regions, and whether or not regions can overlap, and completely classify tractable and intractable cases.
\end{abstract}

 \section{Introduction }
An instance of the {\em Hospitals/Residents problem{}} (\textsc{HR}{} for short) is a many-to-one generalization
of the stable marriage problem.
Its instance consists of a set of {\em residents} and a set of {\em hospitals},
each of whom has a {\em preference list} that strictly orders a subset of the members of the other side.
Furthermore, each hospital has a {\em capacity}, which specifies the maximum number of residents it can accommodate.
A {\em matching} is an assignment of residents to hospitals in which the number of residents assigned to each hospital does not exceed its capacity.
A {\em blocking pair} of a matching $M$ is a pair of a resident $r$ and a hospital $h$, where assigning $r$ to $h$ improves the situation of both  $r$ and $h$.
A matching that admits no blocking pair is a {\em stable matching}.
Gale and Shapley \cite{GS62} showed that a stable matching exists in any \textsc{HR}{} instance, and proposed a polynomial-time algorithm to find one, which is known as the {\em Gale-Shapley algorithm}.

\textsc{HR}{} is quite popular in practical assignment systems, such as assigning students to high schools \cite{APR05,APRS05} or assigning residents to hospitals \cite{Rot84}.
One of the major issues in the residents-hospitals assignment is imbalance in the number of residents assigned to hospitals.
For example, hospitals in urban areas are generally more popular and hence likely to be assigned more residents than those in rural areas.
In fact, it is reported that rural hospitals in Japan suffer from shortage of doctors \cite{KK10,KK12}.
One of the solutions to resolve this problem might be to find a stable matching that assigns more residents to rural hospitals than one found by a frequently used algorithm, such as the Gale-Shapley algorithm.
However, this approach fails since the number of residents assigned to each hospital is invariant over stable matchings, due to the famous {\em Rural Hospital theorem} \cite{GS85,Rot84,Rot86}.

There have been some attempts to resolve this problem by reformulating \textsc{HR}{}.
One is to set a {\em lower quota} to each hospital, which denotes a minimum number of residents required by that hospital.
It is easy to see that, by the Rural Hospitals theorem, there may exist no stable matching that satisfies all the lower quotas.
In some model, a hospital is allowed to be closed if the number of residents cannot reach its minimum quota \cite{BFIM10}, or in another model, it is mandatory to satisfy all the lower quotas, while stability is considered as a soft constraint \cite{HIM16}.
Another formulation is to set a {\em regional cap} to a region consisting of some hospitals, and require the total number of residents assigned to hospitals in the region to be at most its regional cap \cite{BFIM10, KK10}.
We denote this problem the {\em HR with Regional Caps{}} (\textsc{HRRC}{} for short).
We can expect to reduce imbalance in the number of residents by carefully setting the regional caps to urban areas.
However, again due to the Rural Hospital theorem, a stable matching satisfying all the regional caps may not exist.

Kamada and Kojima \cite{KK10,KK15} defined several stability notions for \textsc{HRRC}{}; among them, the most natural one is the {\em strong stability} \cite{KK15}.
They presented an \textsc{HRRC}{}-instance that admits no strongly stable matching, but the computational complexity of the problem of determining the existence of a strongly stable matching had been open up to recently.
In 2020, Aziz et al.~\cite{ABB20} resolved this problem by showing that it is NP-complete.
This hardness holds even if any two regions are disjoint (in other words, each hospital is contained in at most one region) and the size of each region is at most two (i.e., each region contains at most two hospitals).
In their reduction, the length of preference lists are unbounded on both sides, hence one of possible next steps might be to investigate if the problem becomes tractable when it is bounded.

\subsection{Our Contributions}
In this paper, we refine Aziz et al.'s work \cite{ABB20} by investigating the computational complexity of the problem in various restrictions.
Let \textsc{Strong-$(\alpha,\beta,\gamma)$-HRRC} denote the problem of determining if a strongly stable matching exists in an \textsc{HRRC}{} instance where the length of a preference list of each resident is at most $\alpha$, the length of a preference list of each hospital is at most $\beta$, and the size of each region is at most $\gamma$.
We use $\infty$ to mean that the corresponding parameter is unbounded.
Next, let \textsc{HRRCDR}{} be the restriction of \textsc{HRRC}{} where any two regions are disjoint, and \textsc{Strong-$(\alpha,\beta,\gamma)$-HRRCDR}{} be defined analogously as for \textsc{HRRC}{}.
Note then that the above mentioned hardness result by Aziz et al.~\cite{ABB20} can be rephrased as NP-completeness of \textsc{Strong-$(\infty,\infty,2)$-HRRCDR}{}.

In this paper, we completely determine the computational complexity of the problem in terms of the restrictions mentioned above.
Our results are summarized as follows.
First we give results for \textsc{HRRC}{}.
\begin{itemize}
 \item  \textsc{Strong-$(\infty,\infty,1)$-HRRC} is in P. (\cref{L0})
 \item  \textsc{Strong-$(1,\infty,\infty)$-HRRC} is in P. (\cref{L1})
 \item  \textsc{Strong-$(\infty,1,\infty)$-HRRC} is in P. (\cref{L2})
 \item  \textsc{Strong-$(2,2,2)$-HRRC} is NP-complete. (\cref{L3})
\end{itemize}
Hence the problem is tractable if at least one parameter is 1,
while it is intractable even if all the parameters are 2.
We remark that, for all the positive results
(\cref{L0,L1,L2}),
we have also shown that a strongly stable matching always exists.
Next, we give results for \textsc{HRRCDR}{}.
\begin{itemize}
 \item \textsc{Strong-$(\infty,\infty,1)$-HRRCDR}{} is in P. (\cref{L4})
 \item \textsc{Strong-$(1,\infty,\infty)$-HRRCDR}{} is in P. (\cref{L5})
 \item \textsc{Strong-$(\infty,1,\infty)$-HRRCDR}{} is in P. (\cref{L6})
 \item \textsc{Strong-$(2,2,2)$-HRRCDR}{} is in P. (\cref{L7})
 \item \textsc{Strong-$(2,2,3)$-HRRCDR}{} is NP-complete. (\cref{L8})
 \item \textsc{Strong-$(2,3,2)$-HRRCDR}{} is NP-complete. (\cref{L9})
 \item \textsc{Strong-$(3,2,2)$-HRRCDR}{} is NP-complete. (\cref{L10})
\end{itemize}
\cref{L4,L5,L6}
are immediate from
\cref{L0,L1,L2},
respectively, since \textsc{HRRCDR}{} is a special case of \textsc{HRRC}{}.
In contrast to \textsc{HRRC}{}, the problem is tractable if all the parameters are 2, but it becomes intractable if one parameter is raised to 3.
Note that
\cref{L10,L9}
are strengthening of Aziz et al.'s NP-completeness \cite{ABB20}.

  \subsection{Related Work}
Besides strong stability, Kamada and Kojima \cite{KK10,KK15} defined {\em weak stability} and {\em stability} in \textsc{HRRC}{}.
Weak stability is weaker than strong stability and stability lies between them.
They showed that there always exists a weakly stable matching \cite{KK17}, but the proof is not constructive.
For stability, they showed that a stable matching always exists and can be found by a strategyproof algorithm when the set of regions forms a laminar structure \cite{KK18}, i.e, for any two regions, either they are disjoint or one is contained in the other.

Bir{\'{o}} et al.~\cite{BFIM10} proposed a restriction of \textsc{HRRC}{},
where each region has a common preference list
and a preference list of any hospital in the region must be a result of deleting some residents from the common preference list.
They proved that the problem of determining existence of a stable matching is NP-complete even in a very restrictive case.
They also showed that the problem can be solved in polynomial time if the set of regions is laminar.

In Goto et al.~\cite{GIK+16}'s model,
preference lists are complete and each region has a lower quota
as well as an upper quota.
They presented two strategyproof algorithms
to find a fair matching when the set of regions is laminar.
Here, fairness is a weakened notion of the stability in \textsc{HR}{},
where certain types of blocking pairs are allowed to exist.

 \section{Preliminaries}

  \subsection{Definitions for Stable Matching Problems }

{

\newcommand{\hq}[1]{q\hifempty{#1}{}{(#1)}} %

\newcommand{\hc}[1]{c\hifempty{#1}{}{(#1)}} %

\newcommand{\hpreflist}[2]{
\begin{tabular}{rccccccccccccccccccccccc}
 #1
\end{tabular}
\hspace{1cm}
\begin{tabular}{rccccccccccccccccccccccc}
 #2
\end{tabular}
}
   An instance $I = (R, H, \succ_{}, \hq{})$ of {\em Hospitals/Residents problem{}} (or {\em \textsc{HR}{}} for short) consists of
   a set $R$ of {\em residents},
   a set $H$ of {\em hospitals},
   a collection $\succ_{}$ of {\em preference lists},
   and
   a mapping $\hq{}: H\to\mathbb{N}_0$ that represents {\em capacities} of hospitals,
   where $\mathbb{N}_0$ is the set of all nonnegative integers.
   The {\em agents} in $I$ are
   the residents and hospitals in $R \cup H$. %
   Each agent $a \in R \cup H$ has a strictly ordered {\em preference list},
   denoted by $\succ_{a}$,
   over a subset of the members of the opposite set.
   If an agent $a$ prefers $x$ to $y$, we write $x \succ_{a} y$.
   If an agent $x$ is included in an agent $a$'s preference list,
   we say that $x$ is {\em acceptable} to $a$. %
   If $r\in R$ and $h\in H$ are acceptable to each other,
   $(r,h)$ is said to be an {\em acceptable pair}. %
   We assume that acceptability is mutual, that is,
   $r$ is acceptable to $h$ if and only if $h$ is acceptable to $r$. %
   Each hospital $h\in H$ has a nonnegative integral {\em capacity},
   denoted by $\hq{h}$.

   An {\em assignment} $M$ is a subset of $R\times H$. %
   If $(r,h) \in M$, $r$ is said to be {\em assigned to} $h$,
   and $h$ is {\em assigned} $r$. %
   For each $a \in R\cup H$,
   the set of assignees of $a$ in $M$ is denoted by $M(a)$. %
   For $S \subseteq H$, we write
   $M(S) := \bigcup_{h\in S}M(h)$.
   For $r\in R$, if $M(r) = \emptyset$,
   $r$ is said to be {\em unassigned}, otherwise $r$ is {\em assigned}. %
   A hospital $h\in H$ is
   {\em undersubscribed}, or {\em full} according as
   $|M(h)|$ is less than, or equal to $\hq{h}$, respectively. %
   A hospital $h\in H$ is said to be {\em empty} if $M(h)=\emptyset$.
   For notational convenience, given an assignment $M$
   and a resident $r\in R$ such that
   $|M(r)| = 1$,
   we also use $M(r)$ to refer to the unique hospital
   when there is no ambiguity. %
   A {\em matching} $M$ is an assignment
   that satisfies all of the following conditions: %
   (i) for any $(r,h)\in M$, $(r,h)$ is an acceptable pair;
   (ii) for any $r\in R$, $|M(r)|\le 1$;
   (iii) for any $h\in H$, $|M(h)|\le \hq{h}$.

}

\begin{definition}
 \label{L11}

\newcommand{\hq}[1]{q\hifempty{#1}{}{(#1)}} %

\newcommand{\hc}[1]{c\hifempty{#1}{}{(#1)}} %

\newcommand{\hpreflist}[2]{
\begin{tabular}{rccccccccccccccccccccccc}
 #1
\end{tabular}
\hspace{1cm}
\begin{tabular}{rccccccccccccccccccccccc}
 #2
\end{tabular}
}
 Let $I = (R, H,\succ_{},\hq{})$ be an instance of \textsc{HR}{} and
 $M$ be a matching of $I$.
 A pair $(r,h) \in R\times H$ is a {\em blocking pair{}} (or {\em BP{}{}} for short) for $M$
 (or $(r,h)$ {\em blocks} $M$)
 if all of the following conditions are satisfied:
 (i) $(r,h)$ is an acceptable pair;
 (ii) either $r$ is unassigned,
 or $r$ prefers $h$ to $M(r)$;
 (iii) either $h$ is undersubscribed,
 or $h$ prefers $r$ to at least one member of $M(h)$.
 If a matching $M$
 admits no BP{}{}, we say that $M$ is {\em stable}.
\end{definition}

{

\newcommand{\hq}[1]{q\hifempty{#1}{}{(#1)}} %

\newcommand{\hc}[1]{c\hifempty{#1}{}{(#1)}} %

\newcommand{\hpreflist}[2]{
\begin{tabular}{rccccccccccccccccccccccc}
 #1
\end{tabular}
\hspace{1cm}
\begin{tabular}{rccccccccccccccccccccccc}
 #2
\end{tabular}
}
 Gale and Shapley showed that every instance of \textsc{HR}{} admits
 a stable matching, and
 proposed
 an $O( |H| |R| )$-time
 algorithm for finding one \cite{GS62}.
 The algorithm is known as the {\em resident-oriented Gale-Shapley{} algorithm} (or {\em RGS{} algorithm} for short).
}

{

\newcommand{\hq}[1]{q\hifempty{#1}{}{(#1)}} %

\newcommand{\hc}[1]{c\hifempty{#1}{}{(#1)}} %

\newcommand{\hpreflist}[2]{
\begin{tabular}{rccccccccccccccccccccccc}
 #1
\end{tabular}
\hspace{1cm}
\begin{tabular}{rccccccccccccccccccccccc}
 #2
\end{tabular}
}
   An instance of
   {\em Hospitals/Residents problem{} with Regional Caps{}}
 (or {\em \textsc{HRRC}{}} for short)
is $I = (R, H, \succ_{}, \hq{}, \mathcal{E}, \hc{})$ where
   $R$, $H$, $\succ_{}$, $\hq{}$ are as before.
   A {\em region} is a non-empty subset of $H$,
   and
   $\mathcal{E} \subseteq 2^H\setminus\emptyset$ is a set of regions.
   Each region $E\in \mathcal{E}$ has a nonnegative integral {\em regional cap},
   denoted by $\hc{E}$.
For an assignment $M$, a region $E\in \mathcal{E}$ is
   {\em deficient{}}, or {\em full{}} according as
   $|M(E)|$ is less than, or equal to $\hc{E}$, respectively.
}

\begin{definition}
 \label{L12}

\newcommand{\hq}[1]{q\hifempty{#1}{}{(#1)}} %

\newcommand{\hc}[1]{c\hifempty{#1}{}{(#1)}} %

\newcommand{\hpreflist}[2]{
\begin{tabular}{rccccccccccccccccccccccc}
 #1
\end{tabular}
\hspace{1cm}
\begin{tabular}{rccccccccccccccccccccccc}
 #2
\end{tabular}
}
 Let $I = (R,H,\succ_{},\hq{},\mathcal{E},\hc{})$ be an instance of \textsc{HRRC}{}
 and $M$ be an assignment of $I$.
 If $|M(E)| \le c(E)$ for all $E \in \mathcal{E}$,
 $M$ is said to be {\em feasible{}{}}, otherwise $M$ is {\em infeasible{}{}}.
\end{definition}

\begin{definition}
 \label{L13}

\newcommand{\hq}[1]{q\hifempty{#1}{}{(#1)}} %

\newcommand{\hc}[1]{c\hifempty{#1}{}{(#1)}} %

\newcommand{\hpreflist}[2]{
\begin{tabular}{rccccccccccccccccccccccc}
 #1
\end{tabular}
\hspace{1cm}
\begin{tabular}{rccccccccccccccccccccccc}
 #2
\end{tabular}
}
 Let $I = (R,H,\succ_{},\hq{},\mathcal{E},\hc{})$ be an instance of \textsc{HRRC}{}.
 A BP{}{} $(r,h)$ for a matching $M$
 is a {\em strong blocking pair{}} (or {\em SBP{}{}} for short)
 for $M$
 if $(r,h)$ satisfies at least one of the following conditions:
 (i) $M \setminus \set{ (r, M(r)) } \cup \set{ (r,h) }$ is feasible{}{};
 (ii) there exists $r'\in M(h)$ such that $r \succ_{h} r'$.
 If a feasible{}{} matching $M$
 admits no SBP{}{}, we say that $M$ is {\em strongly stable}.
\end{definition}
Note that
if $r \in R$ is unassigned in $M$,
$M \setminus\set{(r,M(r))} = M$
since $(r,M(r)) \not\in M$.

{

\newcommand{\hq}[1]{q\hifempty{#1}{}{(#1)}} %

\newcommand{\hc}[1]{c\hifempty{#1}{}{(#1)}} %

\newcommand{\hpreflist}[2]{
\begin{tabular}{rccccccccccccccccccccccc}
 #1
\end{tabular}
\hspace{1cm}
\begin{tabular}{rccccccccccccccccccccccc}
 #2
\end{tabular}
}
As the definition states,
a BP{}{} $(r,h)$ for a matching $M$ is tolerated to exist if
(i) moving $r$ to $h$ violates a regional cap, and
(ii) $h$ has no incentive to accept $r$ by rejecting a resident in $M(h)$.
}

Kamada and Kojima \cite{KK10,KK15} showed that,
in contrast to \textsc{HR}{},
there exists an instance that
does not admit a strongly stable matching{},
as shown below.
We refer to this instance, which plays an important role in our hardness proofs, as $G_2$.
\begin{example}%

\newcommand{\hq}[1]{q\hifempty{#1}{}{(#1)}} %

\newcommand{\hc}[1]{c\hifempty{#1}{}{(#1)}} %

\newcommand{\hpreflist}[2]{
\begin{tabular}{rccccccccccccccccccccccc}
 #1
\end{tabular}
\hspace{1cm}
\begin{tabular}{rccccccccccccccccccccccc}
 #2
\end{tabular}
}
 \label{L14}
 Let
 $R := \set{r_1, r_2}$ be a set of residents,
 $H := \set{h_1, h_2}$ be a set of hospitals, and
 $\mathcal{E} := \set{ \set{h_1, h_2} }$ be a set of regions.
 Here and hereafter,
 we illustrate agent $a$'s preference list by
 putting ``$:$'' just after the agent's name $a$ and
 listing acceptable agents of $a$ in decreasing order of $a$'s preference.
 We also denote $h[x]$ to represent that a hospital $h$'s capacity is $x$.
 $\hc{ \set{h_1,h_2} } := 1$ represents that the regional cap for a
 region $\set{h_1,h_2}$ is $1$.
 The preference lists, capacities of hospitals, and regional caps $\hc{}$ are as follows.
 \begin{center}
\textup{
 \begin{tabular}{rccccccccccccccccccccccc}
  $r_1$:	& $h_1$	& $h_2$	\\
  $r_2$:	& $h_2$	& $h_1$	\\
 \end{tabular}
 \hspace{1cm}
 \begin{tabular}{rccccccccccccccccccccccc}
  $h_1[1]$:	& $r_2$	& $r_1$	\\
  $h_2[1]$:	& $r_1$	& $r_2$	\\
 \end{tabular} \\
 \begin{tabular}{c}
  $\hc{ \set{h_1, h_2} } := 1$
 \end{tabular}
}
\end{center}
\end{example}

Let {\em Hospitals/Residents problem{} with Regional Caps{} and Disjoint Regions{}} (or \textsc{HRRCDR}{} for short) denote the restriction of
\textsc{HRRC}{}
in which any two
regions are disjoint.
Here, we say that two regions $E_i$ and $E_j$ are disjoint if $E_i \cap E_j = \emptyset$.
Let $\mathbb{N}$ be the set of all positive integers.
Given three integers $\alpha\in \mathbb{N}_0$, $\beta\in \mathbb{N}_0$, and $\gamma\in \mathbb{N}$,
let \textsc{$(\alpha,\beta,\gamma)$-HRRC} (respectively $(\alpha,\beta,\gamma)$-\textsc{HRRCDR}{})
denote the restriction of
\textsc{HRRC}{} (respectively \textsc{HRRCDR}{})
in which
each resident's preference list is of length at most $\alpha$,
each hospital's preference list is of length at most $\beta$,
and
each region's size (i.e., the number of hospitals in the region) is at most $\gamma$.
$\alpha = \infty$ or $\beta = \infty$ means that the lengths of preference lists
are unbounded, and
$\gamma = \infty$ means that the sizes of regions are unbounded.

Let
\textsc{Strong-$(\alpha,\beta,\gamma)$-HRRC}
(\textsc{Strong-$(\alpha,\beta,\gamma)$-HRRCDR}, respectively)
denote the problem of determining if a strongly stable matching{} exists in
an \textsc{HRRC}{} instance
(an \textsc{HRRCDR}{} instance, respectively).

  \subsection{SAT and Its Restricted Variants }
We introduce an NP-complete problem \textsc{SAT}{} \cite{Coo71} %
and its restricted variants.
We will use them later in our proofs for NP-completeness.

An instance of \textsc{SAT}{} consists of
a set of variables and a set of clauses.
Each variable takes either true ($1$) or false ($0$).
If $x$ is a variable, then $x$ and its negation $\bar{x}$ are {\em literals}.
A {\em clause} is a disjunction of literals.
A clause is {\em satisfied} if and only if at least one of its literals is $1$.
We say that a \textsc{SAT}{}-instance is {\em satisfiable} if and only if there exists
a $0$/$1$ assignment to variables
that simultaneously satisfies all the clauses.
Such an assignment is called a {\em satisfying assignment}.
\textsc{SAT}{} asks if there exists a satisfying assignment.
\textsc{3-SAT}{} is a restriction of \textsc{SAT}{}, %
where each clause has exactly three literals.
\textsc{3-SAT}{} is also NP-complete \cite{Coo71}.

Another restriction, \textsc{One-In-Three 3-SAT}{}, is also %
an NP-complete problem \cite{Sch78}.
Its instance consists of
a set $U$ of variables and
a set $C$ of clauses over $U$
such that each clause in $C$
has exactly three literals.
\textsc{One-In-Three 3-SAT}{} asks if
there exists a satisfying assignment such that
each clause in $C$ has exactly one true literal.
\textsc{One-In-Three 3-SAT}{} remains NP-complete even if
no clause in $C$ contains a negated literal \cite{Sch78}.
We refer to \textsc{One-In-Three 3-SAT}{} with such restricted instances as \textsc{One-In-Three Positive 3-SAT}{}. %

 \section{\textsc{HRRC}{} with Intersecting Regions }
 \label{L15}
 In this section,
 we consider the case when
 regions may intersect.
 In this case,
 as we will show in \cref{L0,L1,L2},
 there always exists a strongly stable matching{} in a \textsc{$(\alpha,\beta,\gamma)$-HRRC}-instance
 if at least one of
 $\alpha$, $\beta$, and $\gamma$
 is
 one.
 On the other hand,
 as we will show in \cref{L3},
 determining the existence of a strongly stable matching{} is NP-complete
 even if all of $\alpha$, $\beta$, and $\gamma$ are two.

  \subsection{Polynomial-time Algorithm for \textsc{Strong-$(\infty,\infty,1)$-HRRC} }
\begin{theorem}%

\newcommand{\hq}[1]{q\hifempty{#1}{}{(#1)}} %

\newcommand{\hc}[1]{c\hifempty{#1}{}{(#1)}} %

\newcommand{\hpreflist}[2]{
\begin{tabular}{rccccccccccccccccccccccc}
 #1
\end{tabular}
\hspace{1cm}
\begin{tabular}{rccccccccccccccccccccccc}
 #2
\end{tabular}
}
 \label{L0}
 There exists
 an $O( |H| |R| )$-time
 algorithm to find a
 strongly stable matching, %
 given an \textsc{$(\infty,\infty,1)$-HRRC}-instance $I = (R, H, \succ_{}, \hq{}, \mathcal{E}, \hc{})$.
\end{theorem}
\begin{proof}
{

\newcommand{\hq}[1]{q\hifempty{#1}{}{(#1)}} %

\newcommand{\hc}[1]{c\hifempty{#1}{}{(#1)}} %

\newcommand{\hpreflist}[2]{
\begin{tabular}{rccccccccccccccccccccccc}
 #1
\end{tabular}
\hspace{1cm}
\begin{tabular}{rccccccccccccccccccccccc}
 #2
\end{tabular}
}

 The description of our algorithm is given in \cref{L16}.
\begin{algorithm}[tb]
 \caption{An algorithm for \textsc{Strong-$(\infty,\infty,1)$-HRRC}.}
 \label{L16}
 \begin{calgorithmic}[1]
  \item[\textbf{Input:}] An \textsc{$(\infty,\infty,1)$-HRRC}-instance $I = (R, H, \succ_{}, \hq{}, \mathcal{E}, \hc{})$.
  \item[\textbf{Output:}] A strongly stable matching{} $M$ of $I$.
  \STATE Define a mapping $q': H \to \mathbb{N}_0$ as follows: \\
  $q'(h) := \pcase{ll}{\min( \hq{h}, \hc{E} ) &\mbox{if }\text{there exists $E\in \mathcal{E}$ such that $h\in E$},\\\hq{h}&\mbox{otherwise.}}$
  \label{L17}
  \STATE Let $I' := (H, R, \succ_{}, q')$ be an \textsc{HR}{}-instance.
  \label{L18}
  \STATE Apply RGS{} algorithm to $I'$.
         Let $M$ be the resultant stable matching.
  \RETURN $M$.
 \end{calgorithmic}
\end{algorithm}
 We first show that the output $M$ is a strongly stable matching{} of $I$.
 By the definition of $q'$,
 $q'(h)\le q(h)$ for all $h \in H$.
 Since $M$ is a matching of $I'$,
 $M$ is also a matching of $I$
 because the capacities of hospitals and regional caps are the only difference between $I$ and $I'$.
 For each $E\in \mathcal{E}$ such that $E = \set{h}$,
 $|M(E)| = |M(h)| \le q'(h) = \min(\hq{h},\hc{E}) \le \hc{E}$.
 Therefore,
 (i)
 $M$ is a feasible{}{} matching of $I$.

 Suppose that there exists a BP{}{} for $M$ in $I$.
 Let $(r_*,h_*)$ be one of such BP{}{}s.
 Since $(r_*,h_*)$ is {\em not} a BP{}{} for $M$ in $I'$,
 we have
 $|M(h_*)| = q'(h_*)$ and
 (ii) for any $r'\in M(h_*)$, $r'\succ_{h_*}r_*$.
 Recalling that $(r_*,h_*)$ is a BP{}{} for $M$ in $I$,
 $|M(h_*)| < \hq{h_*}$.
 Thus, $q'(h_*) \ne \hq{h_*}$ holds,
 and hence
 there exists $E_*\in \mathcal{E}$ such that $h_*\in E_*$ and $q'(h_*) = \hc{E_*}$.
 By the assumption,
 $|E_*| = 1$, so
 $E_* = \set{h_*}$ holds.
 Let $M^+ := M \setminus\set{ (r_*,M(r_*)) } \cup \set{ (r_*,h_*) }$.
 Then, $M^+(E_*) = M^+(\set{h_*}) = M(h_*) \cup \set{r_*}$.
 Since $r_*\not\in M(h_*)$, $|M^+(E_*)| = |M(h_*)|+1 = \hc{E_*} + 1$ holds.
 Therefore,
 (iii) $M \setminus\set{ (r_*,M(r_*)) } \cup \set{ (r_*,h_*) }$ is infeasible{}{} in $I$.
 By conditions (i) to (iii),
 $M$ is a strongly stable matching{} of $I$.

 We then show the time complexity.
 Since the size of each region is at most one,
 for any hospital $h$,
 there exits at most one region that contains $h$.
 Therefore, \cref{L17} can be executed in $O( |H| )$ time.
 Since
 all components of $I'$
 have been already constructed,
 \cref{L18} can be executed in $O(1)$ time.
 Since RGS{} algorithm runs in $O( |H| |R| )$ time,
 the whole time complexity is $O( |H| |R| )$.%
}
\end{proof}

  \subsection{Polynomial-time Algorithm for \textsc{Strong-$(1,\infty,\infty)$-HRRC} }
\begin{theorem}
 \label{L1}

\newcommand{\hq}[1]{q\hifempty{#1}{}{(#1)}} %

\newcommand{\hc}[1]{c\hifempty{#1}{}{(#1)}} %

\newcommand{\hpreflist}[2]{
\begin{tabular}{rccccccccccccccccccccccc}
 #1
\end{tabular}
\hspace{1cm}
\begin{tabular}{rccccccccccccccccccccccc}
 #2
\end{tabular}
}

 There exists
 an $O( |H| + |\mathcal{E}| |R| )$-time algorithm
 to find a
 strongly stable matching{}, %
 given a \textsc{$(1,\infty,\infty)$-HRRC}-instance $I = (R, H, \succ_{}, \hq{}, \mathcal{E}, \hc{})$.
\end{theorem}
\begin{proof}
{

\newcommand{\hq}[1]{q\hifempty{#1}{}{(#1)}} %

\newcommand{\hc}[1]{c\hifempty{#1}{}{(#1)}} %

\newcommand{\hpreflist}[2]{
\begin{tabular}{rccccccccccccccccccccccc}
 #1
\end{tabular}
\hspace{1cm}
\begin{tabular}{rccccccccccccccccccccccc}
 #2
\end{tabular}
}

 We show that
 \begin{algorithm}[tb]
  \caption{An algorithm for \textsc{Strong-$(1,\infty,\infty)$-HRRC}.}
  \label{L19}
  \begin{calgorithmic}[1]
   \item[\textbf{Input:}] A \textsc{$(1,\infty,\infty)$-HRRC}-instance $I = (R, H, \succ_{}, \hq{}, \mathcal{E}, \hc{})$.
   \item[\textbf{Output:}] A strongly stable matching{} $M$ of $I$.
   \STATE $M := \emptyset$. \label{L20}
   \FOR{\textbf{each} $h \in H$} \label{L21}
     \WHILE{$h$'s preference list is not empty} \label{L22}
       \STATE Let $r$ be the first resident in $h$'s preference list.
   Remove $r$ from the list. \label{L23}
       \IF{$|M(h)| < \hq{h}$ and
         for any $E\in \mathcal{E}$ such that $h\in E$, $|M(E)| < \hc{E}$}
         \label{L24}
         \STATE $M := M \cup \set{(r,h)}$. \label{L25}
       \ENDIF
     \ENDWHILE
   \ENDFOR
   \RETURN $M$.
  \end{calgorithmic}
 \end{algorithm}
 \cref{L19} always outputs
 a strongly stable matching{} in $O( |H| + |\mathcal{E}| |R| )$ time.
 For each $h\in H$,
 the algorithm greedily assigns the best possible residents
 to $h$ with respect to $h$'s preference,
 as long as capacities of the hospital $h$ and the regions including $h$ are not
 violated.

 We first show the correctness.
 Let $M_1$ be $M$ at the end of the algorithm.
 We will prove that $M_1$ is a strongly stable matching{} of $I$ by
 showing that
 $M_1$ is feasible{}{} in $I$
 and
 $M_1$ admits no SBP{}{}.
 The key observation is the monotonicity of the condition in \cref{L24},
 i.e., for each $h\in H$,
 once the condition is unsatisfied,
 it will never be satisfied again after that.

 We first show that $M_1$ is a feasible{}{} matching of $I$.
 In the algorithm,
 the only operations on $M$ are the initialization to $\emptyset$ in \cref{L20} and
 the append of a pair in \cref{L25}.
 Since the length of each resident's preference list is at most one,
 \cref{L23} implies that \cref{L25} is executed at most once for each $r \in R$.
 Therefore, each resident $r \in R$ appears in $M$ at most once.
 Further, in \cref{L24}, %
 the algorithm checks whether
 appending $(r,h)$ to $M$ violates neither
 $h$'s capacity nor regional cap for any region that includes $h$.
 Thus, $M$ is always feasible{}{},
 and hence $M_1$ is also feasible{}{}.

 Next, we show that
 the following two conditions hold
 for any BP{}{} $(r_*,h_*)$ for $M_1$:
 (i) $M_1 \setminus \set{(r_*,M_1(r_*))} \cup \set{(r_*,h_*)}$ is infeasible{}{}, and
 (ii) for any $r' \in M_1(h_*)$, $r' \succ_{h_*}r_*$,
 which implies that $(r_*,h_*)$ is not an SBP{}{}.
 Note that $M_1(r_*) = \emptyset$ holds,
 since the length of $r_*$'s preference list is exactly one and
 $h_*$ is the unique hospital in $r_*$'s list.

 For contradiction, suppose that (ii) does not hold.
 Then, there exists a resident $r' \in M_1(h_*)$ such that $r_*\succ_{h_*}r'$. %
 This means that at some point $r_*$ was not assigned to $h_*$,
 i.e., the condition in \cref{L24} was not satisfied,
 but later $r'$ was assigned to $h_*$,
 i.e., the condition in \cref{L24} was satisfied.
 This contradicts the monotonicity of the condition in \cref{L24}.

 For contradiction, suppose that (i) does not hold.
 Since $M_1(r_*) = \emptyset$,
 $M_1 \setminus \set{(r_*,M_1(r_*))} \cup \set{(r_*,h_*)} = M_1 \cup \set{(r_*,h_*)}$,
 so $M_1 \cup \set{(r_*,h_*)}$ is feasible{}{}.
 Also,
 $|M_1(h_*)| < \hq{h_*}$ because
 $(r_*,h_*)$ is a BP{}{} for $M_1$ and (ii) holds.
 Thus, the condition in \cref{L24} for $h_*$
 is still satisfied at the end of algorithm.
 However, $(r_*,h_*)$ was not added to $M$,
 i.e., the condition in \cref{L24} for $h_*$ was not satisfied.
 This contradicts the monotonicity of the condition in \cref{L24} again.

 We then analyze the time complexity of \cref{L19}.
 Since the length of each resident's preference list is at most one,
 \cref{L23} is executed at most $|R|$ times.
To check the conditions $|M(h)|<\hq{h}$ and $|M(E)| < \hc{E}$ at \cref{L24},
we prepare a counter for each hospital and each region.
Then, for processing each resident, checking the conditions at \cref{L24}
and updating the counters at \cref{L25} can be done in time $O(|\mathcal{E}|)$.
Clearly updating the matching at \cref{L25} can be done in constant time.
 Therefore, the algorithm runs in $O( |H| + |\mathcal{E}| |R| )$ time.
}
\end{proof}

  \subsection{Polynomial-time Algorithm for \textsc{Strong-$(\infty,1,\infty)$-HRRC} }
\begin{theorem}
 \label{L2}

\newcommand{\hq}[1]{q\hifempty{#1}{}{(#1)}} %

\newcommand{\hc}[1]{c\hifempty{#1}{}{(#1)}} %

\newcommand{\hpreflist}[2]{
\begin{tabular}{rccccccccccccccccccccccc}
 #1
\end{tabular}
\hspace{1cm}
\begin{tabular}{rccccccccccccccccccccccc}
 #2
\end{tabular}
}

 There exists
 an $O(|R|+|\mathcal{E}| |H|)$-time algorithm
 to find a
 strongly stable matching{}, %
 given an \textsc{$(\infty,1,\infty)$-HRRC}-instance $I = (R, H, \succ_{}, \hq{}, \mathcal{E}, \hc{})$.
\end{theorem}
\begin{proof}
{

\newcommand{\hq}[1]{q\hifempty{#1}{}{(#1)}} %

\newcommand{\hc}[1]{c\hifempty{#1}{}{(#1)}} %

\newcommand{\hpreflist}[2]{
\begin{tabular}{rccccccccccccccccccccccc}
 #1
\end{tabular}
\hspace{1cm}
\begin{tabular}{rccccccccccccccccccccccc}
 #2
\end{tabular}
}

 We show that
\begin{algorithm}[tb]
 \caption{An algorithm for \textsc{Strong-$(\infty,1,\infty)$-HRRC}.}
 \label{L26}
 \begin{calgorithmic}[1]
  \item[\textbf{Input:}] An \textsc{$(\infty,1,\infty)$-HRRC}-instance $I = (R, H, \succ_{}, \hq{}, \mathcal{E}, \hc{})$.
  \item[\textbf{Output:}] A strongly stable matching{} $M$ of $I$.
  \STATE $M := \emptyset$. \label{L27}
    \FOR{\textbf{each} $r \in R$} \label{L28}
        \STATE Let $h$ be the first hospital in $r$'s preference list (if any)
  that satisfies both of the following conditions:
  (i)  $|M(h)| < \hq{h}$;
  (ii) for any $E\in \mathcal{E}$ such that $h\in E$, $|M(E)|<\hc{E}$.
  \label{L29}
      \IF{such a hospital $h$ exists}
        \STATE $M := M \cup \set{(r,h)}$. \label{L30}
      \ENDIF
    \ENDFOR
  \RETURN $M$.
 \end{calgorithmic}
\end{algorithm}
 \cref{L26} always outputs
 a strongly stable matching{} in $O( |R| + |\mathcal{E}| |H| )$ time.
 For each $r\in R$,
 the algorithm assigns $r$ to the best possible $h\in H$ with respect to
 $r$'s preference list,
 as long as the capacity of $h$ and feasibility are not violated.

We first show the correctness.
Let $M_{1}$ be $M$ at the end of the algorithm.
It is easy to see that a matching $M$ is always
feasible{}{} during the course of the algorithm,
since $(r, h)$ is added to $M$ only when
this addition does not violate feasibility.
Hence $M_{1}$ is also feasible{}{}.

Next, let $(r_*, h_*)$ be any BP{}{} of $M_{1}$ (if any).
Then $r_*$ is either unassigned or assigned to a hospital worse than $h_*$ in $M_{1}$.
Therefore, when the \textbf{for}-loop was executed for $r_*$,
$h_*$ was considered but did not satisfy the conditions in \cref{L29}.
Since each hospital contains at most one resident in the preference list,
$r_*$ is the only resident in $h_*$'s list.
It must be the case that $\hq{h_*}>0$, as otherwise, $(r_*, h_*)$ cannot be a BP{}{}.
Thus condition (i) in \cref{L29} must be satisfied for $r_*$ and $h_*$,
so condition (ii) was unsatisfied.
This means that at this moment there is a region $E_*$
that includes $h_*$ such that $|M(E_*)| = \hc{E_*}$,
and this equation is retained until the end of the algorithm,
implying that $|M_{1}(E_*)| = \hc{E_*}$.
Note that if $M_{1}(r_*)=h'$, $h'$ is not in $E_*$.
Therefore $M_{1} \setminus \set{(r_*,M_{1}(r_*))} \cup \set{(r_*,h_*)}$
violates $\hc{E_*}$ and hence is not feasible{}{},
implying that $(r_*, h_*)$ is {\em not} an SBP{}{}.
This completes the proof that $M_{1}$ is strongly stable.

 We then analyze the time complexity of \cref{L26}.
 Since the length of each hospital's preference list is at most one,
 testing the conditions in \cref{L29}
 is executed at most $|H|$ times.
To check the conditions $|M(h)|<\hq{h}$ and $|M(E)| < \hc{E}$ at \cref{L29},
we prepare a counter for each hospital and each region.
Then, for processing each hospital, checking the conditions at \cref{L29}
and updating the counters at \cref{L30} can be done in time $O(|\mathcal{E}|)$.
Clearly updating the matching at \cref{L30} can be done in constant time.
Therefore, the algorithm runs in $O( |R| + |\mathcal{E}| |H| )$ time.
}
\end{proof}

  \subsection{NP-completeness of \textsc{Strong-$(2,2,2)$-HRRC} }
\begin{theorem}%
 \label{L3}

 \textsc{Strong-$(2,2,2)$-HRRC} is NP-complete.
\end{theorem}
\begin{proof}

\newcommand{\hq}[1]{q\hifempty{#1}{}{(#1)}} %

\newcommand{\hc}[1]{c\hifempty{#1}{}{(#1)}} %

\newcommand{\hpreflist}[2]{
\begin{tabular}{rccccccccccccccccccccccc}
 #1
\end{tabular}
\hspace{1cm}
\begin{tabular}{rccccccccccccccccccccccc}
 #2
\end{tabular}
}

Membership in NP is obvious.
We show show a reduction from \textsc{One-In-Three Positive 3-SAT}{}.
 Let $I$ be an instance of \textsc{One-In-Three Positive 3-SAT}{} having
 $n$ variables $x_{i}$ ($i \in [1, n]$) and
 $m$ clauses $C_{j}$ ($j \in [1, m]$) of size three.
 We construct an \textsc{HRRC}{}-instance $I' := (R, H, \succ_{}, \hq{}, \mathcal{E}, \hc{})$ %
 from $I$.
 For each variable $x_{i}$, we construct a {\em variable gadget{}},
 which consists of a resident $y'_{i}$ and a hospital $x'_{i}$.
 A variable gadget{} corresponding to $x_{i}$ is called an {\em $x_{i}$-gadget}.
 For each clause $C_{j}$, we construct a {\em clause gadget{}}.
 It contains
 two residents $g'_{j,1}$ and $g'_{j,3}$, and
 two hospitals $g'_{j,2}$ and $g'_{j,4}$.
 A clause gadget{} corresponding to $C_{j}$ is called a {\em $C_{j}$-gadget}.
 Thus, there are $n + 2m$ residents and $n + 2m$ hospitals in the
 created \textsc{HRRC}{}-instance, denoted $I'$.

 \begin{figure}[tbp]
   \centering
   \hpreflist{
   $y'_{i}$:  & $x'_{i}$ \\
   }{
   $x'_{i}[1]$:  & $y'_{i}$ \\
   }
   \caption{Preference lists and capacities of agents in $x_{i}$-gadget.}
   \label{L31}
 \end{figure}
 \begin{figure}[tbp]
  \centering
  \hpreflist{
  $g'_{j,1}$: & $g'_{j,2}$ & $g'_{j,4}$ \\
  $g'_{j,3}$: & $g'_{j,4}$ & $g'_{j,2}$ \\
  }{
  $g'_{j,2}[1]$: & $g'_{j,3}$ & $g'_{j,1}$ \\
  $g'_{j,4}[1]$: & $g'_{j,1}$ & $g'_{j,3}$ \\
  }
  \caption{Preference lists and capacities of agents in $C_{j}$-gadget.}
  \label{L32}
 \end{figure}
 The preference lists and capacities of agents
 in the $x_{i}$-gadget and the $C_{j}$-gadget are
 constructed as shown in \cref{L31,L32}, respectively.
 Suppose that $C_{j}$'s $k$th literal ($k\in \set{1,2,3}$) is $x_{ i_{j,k} }$.
 For each clause $C_{j} = \{ x_{ i_{j,1} } \vee x_{ i_{j,2} } \vee x_{ i_{j,3} } \}$,
 we create regions of size two as follows:
 for each $E$ such that
 $E \subset\set{g'_{j,2}, g'_{j,4}, x'_{ i_{j,1} }, x'_{ i_{j,2} }, x'_{ i_{j,3} }}$
 and $|E| = 2$,
 we add $E$ to the set of regions $\mathcal{E}$
 and set the regional cap of $E$ as $\hc{E}:=1$.
 Now the reduction is completed.
 It is not hard to see that the reduction
 can be performed in polynomial time and the
 conditions on the preference lists and regions
 stated in the theorem are satisfied.

 We then show the correctness.
 First, suppose that $I$ is satisfiable and
 let $A$ be a satisfying assignment.
 We construct a strongly stable matching{} $M$ of
 $I'$ from $A$.
 For an $x_{i}$-gadget,
 add $(y'_{i}, x'_{i})$ to $M$
 if and only if
 $x_{i} = 1$ under $A$.
 No agent in $C_{j}$-gadget is assigned.
 The construction of $M$ is now completed.

 We show that $M$ is a feasible{}{} matching of $I'$.
 By the construction, $M$ is a matching of $I'$.
 For a clause $C_{j} = \{ x_{ i_{j,1} } \vee x_{ i_{j,2} } \vee x_{ i_{j,3} } \}$,
 exactly one literal is $1$ under $A$.
 Therefore, in $M$, exactly one of $x'_{ i_{j,1} }$, $x'_{ i_{j,2} }$, and $x'_{ i_{j,3} }$
 is assigned its first-choice resident.
 Also, neither $g'_{j,2}$ nor $g'_{j,4}$ is assigned in $M$.
 Therefore, all the regional caps are satisfied.
 Thus, $M$ is feasible{}{} in $I'$.

 Next, we show that $M$ has no SBP{}{}. %
 All hospitals that are nonempty in $M$ are assigned their first-choice residents,
 so none of them is a part of an SBP{}{}.
 Let $x'_{i}$ be an empty hospital and
 $C_{j} = \{ x_{ i_{j,1} } \vee x_{ i_{j,2} } \vee x_{ i_{j,3} } \}$
 be a clause that contains $x_{i}$.
 By the satisfiability of $A$, there exists $\ell$
 such that $x_{ i_{j,\ell} } = 1$, but $i_{j,\ell}\ne i$ because $x_{i}=0$.
 Then, $x'_{ i_{j,\ell} }$ is assigned its first-choice resident in $M$.
 Since $\hc{ \set{x'_{i}, x'_{ i_{j,\ell} }} } = 1$,
 $x'_{i}$ is not a part of an SBP{}{}.
 For hospitals $g'_{j,2}$ and $g'_{j,4}$,
 there exists a hospital $x'_{ i_{j,\ell} }$ that is assigned its first-choice resident.
 Since $\hc{ \set{ g'_{j,2}, x'_{ i_{j,\ell} } } } = 1$ and $\hc{ \set{ g'_{j,4}, x'_{ i_{j,\ell} } } } = 1$,
 neither $g'_{j,2}$ nor $g'_{j,4}$ is a part of an SBP{}{}.
 Thus, $M$ is a strongly stable matching{}.

 Conversely, suppose that $I'$ admits a strongly stable matching{} $M$.
 We construct a satisfying assignment $A$ of $I$.
 For each $x'_{i}$, if $M(x'_{i}) \ne \emptyset$ then we set $x_{i} = 1$ in $A$;
 otherwise, we set $x_{i} = 0$ in $A$.

 We show that $A$ satisfies $I$.
 Suppose not.
 Then, there exists a clause $C_{j} = \{ x_{ i_{j,1} } \vee x_{ i_{j,2} } \vee x_{ i_{j,3} } \}$
 such that either (i) none of the literals are $1$, or (ii) at least two literals are $1$.
 In the case (i),
 all of $x'_{ i_{j,1} }$, $x'_{ i_{j,2} }$, and $x'_{ i_{j,3} }$ are empty.
 Therefore, regional caps 
 $\hc{ \set{ g'_{j,2}, x'_{ i_{j,\ell} } } } = 1$ and
 $\hc{ \set{ g'_{j,4}, x'_{ i_{j,\ell} } } } = 1$ ($\ell \in [1,3]$)
 do not affect the matching among
 $g'_{j,1}$, $g'_{j,2}$, $g'_{j,3}$, and $g'_{j,4}$.
 Then, these four agents form the same structure as $G_2$ in \cref{L14},
 so there exists an SBP{}{}
 no matter how they are assigned,
 a contradiction.
 In the case (ii),
 let $x_{ i_{j,\ell_1} }$ and $x_{ i_{j,\ell_2} }$ be
 two literals that take a value of $1$.
 Then, both $M( x'_{ i_{j,\ell_1} } ) \ne \emptyset$
 and $M( x'_{ i_{j,\ell_2} } ) \ne \emptyset$ hold.
 Since they have a regional cap of $\hc{ \set{ x'_{ i_{j,\ell_1} }, x'_{ i_{j,\ell_2} } } } = 1$,
 $M$ is infeasible{}{}, a contradiction.
 Thus, $A$ satisfies $I$.
\end{proof}

 \section{\textsc{HRRC}{} with Disjoint Regions }
 \label{L33}
In this section, we consider the case when regions are disjoint.
 Since \textsc{HRRCDR}{} is a special case of \textsc{HRRC}{},
 the following corollaries are immediate from \cref{L0,L1,L2}.
\begin{corollary}
  \label{L4}
  \textsc{Strong-$(\infty,\infty,1)$-HRRCDR} is in P.
 \end{corollary}
 \begin{corollary}
  \label{L5}
  \textsc{Strong-$(1,\infty,\infty)$-HRRCDR} is in P.
 \end{corollary}
 \begin{corollary}
  \label{L6}
  \textsc{Strong-$(\infty,1,\infty)$-HRRCDR} is in P.
 \end{corollary}
 Hence, in the rest of this section, we consider cases where each parameter is at least two.
We first show a positive result in \cref{L34}.
Then, in \cref{L35}, we introduce an NP-complete variant of \textsc{SAT}{},
which will be used in the negative results in \cref{L36,L37,L38}.

  \subsection{Polynomial-time Algorithm for \textsc{Strong-$(2,2,2)$-HRRCDR} }
  \label{L34}
{
In this section, we give a polynomial-time algorithm for
\textsc{Strong-$(2,2,2)$-HRRCDR}.
A high-level idea of the algorithm is as follows.
Given an instance $I$, we first divide it to subinstances.
Let a {\em $(2 \times 2)$-subinstance} be a subinstance of $I$
consisting of two residents $r_{1}, r_{2}$ and two hospitals $h_{1}, h_{2}$
such that each $r_{i}$ ($i=1, 2$) includes both $h_{1}$ and $h_{2}$ in the preference list,
each $h_{i}$ ($i=1, 2$) includes both $r_{1}$ and $r_{2}$ in the preference list,
and $h_{1}$ and $h_{2}$ are included in the same region.
Since regions are disjoint, a $(2 \times 2)$-subinstance has no interference with the rest of the instance.
Suppose that $I$ contains $p$ $(2 \times 2)$-subinstances $I_{1}, I_{2}, \ldots, I_{p}$.
Remove all of them from $I$ and let $I_{0}$ be the resultant subinstance.

We then handle these subinstances independently.
Each $(2 \times 2)$-subinstance is solved by a brute-force search.
Before solving the subinstance $I_{0}$, to simplify the analysis, we apply the {\em shrinking} operation described in \cref{L39}.
It modifies capacities of hospitals but does not change the set of strongly stable matchings.
Then, $I_{0}$ is solved by a sub-algorithm \cref{L40}, which will be provided in \cref{L41}.
It is guaranteed that $I_{0}$ admits a strongly stable matching, and \cref{L40} returns one, say $M_{0}$.
If there is a $(2 \times 2)$-subinstance that admits no strongly stable matching, we conclude that $I$ also admits none.
If all the $(2 \times 2)$-subinstances admit a stable matching, then their union plus $M_{0}$ is a strongly stable matching of $I$.
In \cref{L42}, we give the whole algorithm and show its correctness.
}

\subsubsection{Notation and Shrinking Operation for \textsc{HRRC}{}-instances}
\label{L39}
{

\newcommand{\hq}[1]{q\hifempty{#1}{}{(#1)}} %

\newcommand{\hc}[1]{c\hifempty{#1}{}{(#1)}} %

\newcommand{\hpreflist}[2]{
\begin{tabular}{rccccccccccccccccccccccc}
 #1
\end{tabular}
\hspace{1cm}
\begin{tabular}{rccccccccccccccccccccccc}
 #2
\end{tabular}
}
 To simplify the description,
 we introduce two symbols, $\funcdef{AC}{}$ and $\funcdef{CommonR}{}$.
}
\begin{definition}

\newcommand{\hq}[1]{q\hifempty{#1}{}{(#1)}} %

\newcommand{\hc}[1]{c\hifempty{#1}{}{(#1)}} %

\newcommand{\hpreflist}[2]{
\begin{tabular}{rccccccccccccccccccccccc}
 #1
\end{tabular}
\hspace{1cm}
\begin{tabular}{rccccccccccccccccccccccc}
 #2
\end{tabular}
}

 Let $I = (R, H, \succ_{}, \hq{}, \mathcal{E}, \hc{})$ be an instance of \textsc{HRRC}{}.
 For an agent $a$,
 we define
 $\funcdef{AC}{a}$
 to be the set of agents that are acceptable to $a$.
 For $H' \subset H$,
 we define
 $\funcdef{CommonR}{H'} := \bigcap_{h \in H'}\funcdef{AC}{h}$
 to be the set of residents that are acceptable to all hospitals in $H'$.
\end{definition}
We then introduce an operation called {\em shrinking}.
Given an instance of \textsc{HRRC}{},
it reduces redundant hospital capacity without changing the set of strongly stable matching{}s of the instance.
\begin{definition}
 \label{L43}

\newcommand{\hq}[1]{q\hifempty{#1}{}{(#1)}} %

\newcommand{\hc}[1]{c\hifempty{#1}{}{(#1)}} %

\newcommand{\hpreflist}[2]{
\begin{tabular}{rccccccccccccccccccccccc}
 #1
\end{tabular}
\hspace{1cm}
\begin{tabular}{rccccccccccccccccccccccc}
 #2
\end{tabular}
}
 Let $I = (R, H, \succ_{}, \hq{}, \mathcal{E}, \hc{})$ be an instance of \textsc{HRRC}{}.
 Let $q': H \to \mathbb{N}_0$ be a mapping defined by
 $q'(h) := \min\set{\hq{h}, |\funcdef{AC}{h}|}$.
 Then, an instance $I' := (R, H, \succ_{}, q', \mathcal{E}, \hc{})$ of \textsc{HRRC}{}
 is called a shrinked instance of $I$.
 This operation is denoted by $\funcdef{Shrink}{}$.
 That is, we write $I' = \funcdef{Shrink}{I}$.
\end{definition}
\begin{lemma}
 \label{L44}

\newcommand{\hq}[1]{q\hifempty{#1}{}{(#1)}} %

\newcommand{\hc}[1]{c\hifempty{#1}{}{(#1)}} %

\newcommand{\hpreflist}[2]{
\begin{tabular}{rccccccccccccccccccccccc}
 #1
\end{tabular}
\hspace{1cm}
\begin{tabular}{rccccccccccccccccccccccc}
 #2
\end{tabular}
}
 Given an \textsc{HRRC}{}-instance $I$,
 $M$ is a strongly stable matching{} of $I$
 if and only if
 $M$ is a strongly stable matching{} of $\funcdef{Shrink}{I}$.
\end{lemma}
\begin{proof}
{

\newcommand{\hq}[1]{q\hifempty{#1}{}{(#1)}} %

\newcommand{\hc}[1]{c\hifempty{#1}{}{(#1)}} %

\newcommand{\hpreflist}[2]{
\begin{tabular}{rccccccccccccccccccccccc}
 #1
\end{tabular}
\hspace{1cm}
\begin{tabular}{rccccccccccccccccccccccc}
 #2
\end{tabular}
}

 Let $I := (R, H, \succ_{}, \hq{}, \mathcal{E}, \hc{})$ be an \textsc{HRRC}{}-instance and
 $I' := (R, H, \succ_{}, q', \mathcal{E}, \hc{})$ be a shrinked instance of $I$.
 First, suppose that $M$ is a strongly stable matching{} of $I$.
 Since $|M(h)| \le \hq{h}$ and $|M(h)| \le |\funcdef{AC}{h}|$ for any $h\in H$,
 by the definition of $q'$,
 $|M(h)|\le q'(h)$ for any $h\in H$.
 Thus, $M$ is a matching of $I'$.
 Further, since $I$ and $I'$ have the same regional caps
 and $M$ is feasible{}{} in $I$, %
 $M$ is also feasible{}{} in $I'$.
 Let $(r_*, h_*)$ be any BP{}{} for $M$ in $I'$ (if any).
 Since $I$ and $I'$ have the same preference lists
 and $q'(h)\le \hq{h}$ for any $h\in H$,
 $(r_*, h_*)$ is also a BP{}{} for $M$ in $I$.
 Since $M$ is strongly stable in $I$,
 $(r_*, h_*)$ satisfies
 neither condition (i) nor (ii) of \cref{L13} in $I$.
 This is also true in $I'$,
 since $I$ and $I'$ have the same preference lists and regional caps.
 Then, $(r_*, h_*)$ is not an SBP{}{} for $M$ in $I'$.
 Thus, $M$ is also a strongly stable matching{} of $I'$.

 Conversely,
 suppose that $M$ is a strongly stable matching{} of $I'$.
 Since the $\funcdef{Shrink}{}$ operation only reduces the capacity of hospitals,
 it is straightforward to verify that $M$ is also a feasible{}{} matching of $I$.
 Let $(r_*, h_*)$ be any BP{}{} for $M$ in $I$ (if any).
 Then, $r_*$ is unassigned or prefers $h_*$ to $M(r_*)$.
 There are two cases for $h_*$:
 $|M(h_*)|<\hq{h_*}$,
 or
 there exists $r'\in M(h_*)$ such that $r_* \succ_{h_*} r'$.
 In the former case,
 $|M(h_*)|<|\funcdef{AC}{h_*}|$ also holds
 since $(r_*, h_*)\not\in M$.
 By the definition of $q'$,
 we have $|M(h_*)|<q'(h_*)$,
 and hence $(r_*, h_*)$ blocks $M$ in $I'$.
 In the latter case,
 $r_* \succ_{h_*} r'$ also holds in $I'$
 since $I$ and $I'$ have the same preference lists.
 Therefore, $(r_*, h_*)$ blocks $M$ in $I'$.
 Thus, in either case, $(r_*, h_*)$ blocks $M$ in $I'$.
 Since $M$ is strongly stable in $I'$,
 $(r_*, h_*)$ satisfies
 neither condition (i) nor (ii) of \cref{L13} in $I'$.
 This is also true in $I$,
 since $I$ and $I'$ have the same preference lists and regional caps.
 Then, $(r_*, h_*)$ is not an SBP{}{} for $M$ in $I$.
 Thus, $M$ is also a strongly stable matching{} of $I$.
}
\end{proof}

\subsubsection{Algorithm for $(2 \times 2)$-free Subinstances}
\label{L41}
Here we propose an algorithm to solve $(2 \times 2)$-free subinstances.
Before describing the algorithm,
we introduce a property of \textsc{HR}{}, which will be
used in the correctness proof.
In \textsc{HR}{}, each hospital is assigned the same number of residents
in all the stable matchings \cite{Rot84,GS85,Rot86}.  %
The following corollary, which can be easily derived from the proof of Lemma 9 in \cite{HIM16},
states how the number of assigned residents changes if 
the capacity of a hospital is slightly decreased.
\begin{corollary}
 \label{L45}
 Let $I_0$ be an instance of \textsc{HR}{},
 $H$ be the set of hospitals,
 and $h\in H$ be any hospital with positive capacity.
 Let $I_1$ be equivalent to $I_0$ except that only the capacity
 of $h$ is decreased by $1$.
 For each $i \in \set{0,1}$, let $M_i$ be a stable matching of $I_i$.
 Then,
 (i) $ |M_1(h)| \ge |M_0(h)|-1$,
 and
 (ii) for any $h' \in H \setminus \set{h}$, $|M_1(h')| \ge |M_0(h')|$.
\end{corollary}

{

\newcommand{\hq}[1]{q\hifempty{#1}{}{(#1)}} %

\newcommand{\hc}[1]{c\hifempty{#1}{}{(#1)}} %

\newcommand{\hpreflist}[2]{
\begin{tabular}{rccccccccccccccccccccccc}
 #1
\end{tabular}
\hspace{1cm}
\begin{tabular}{rccccccccccccccccccccccc}
 #2
\end{tabular}
}
 Now, the algorithm is shown in \cref{L40}.
 It accepts an \textsc{HRRC}{}-instance $(R, H, {\succ_{},} \hq{}, \mathcal{E}, \hc{})$ that satisfies all of the following properties:
 $\mathcal{E}$ is disjoint;
 for any $r\in R$, $|\funcdef{AC}{r}|\le2$;
 for any $h\in H$, $|\funcdef{AC}{h}|\le2$;
 for any $E\in \mathcal{E}$, $|E|\le2$;
 for any $E\in \mathcal{E}$ such that $|E|=2$, $|\funcdef{CommonR}{E}|\le1$;
 and
 for any $h\in H$, $\hq{h}\le2$.
 Note that a region such that $|E|=2$ and $|\funcdef{CommonR}{E}| = 2$ corresponds to a $(2 \times 2)$-subinstance in a given instance.
Hence if there is no $(2 \times 2)$-subinstance,
for any $E\in \mathcal{E}$ such that $|E|=2$,
$|\funcdef{CommonR}{E}|\le1$ is guaranteed.

\cref{L40} finds a stable matching by using RGS{} algorithm.
As long as the output stable matching is infeasible, it chooses one region whose regional cap is broken, chooses one hospital $h'$ in it, and reduces its capacity by one.
When a feasible matching is found, \cref{L40} outputs it. 
}
\begin{algorithm}[tb]

\newcommand{\hq}[1]{q\hifempty{#1}{}{(#1)}} %

\newcommand{\hc}[1]{c\hifempty{#1}{}{(#1)}} %

\newcommand{\hpreflist}[2]{
\begin{tabular}{rccccccccccccccccccccccc}
 #1
\end{tabular}
\hspace{1cm}
\begin{tabular}{rccccccccccccccccccccccc}
 #2
\end{tabular}
}

 \caption{Finding a strongly stable matching{} for $(2 \times 2)$-free shrinked subinstance.}
 \label{L40}
 \begin{calgorithmic}[1]
  \item[\textbf{Input:}] A $(2,2,2)$-\textsc{HRRCDR}{}-instance $I = (R, H, \succ_{}, \hq{}, \mathcal{E}, \hc{})$.
  \REQUIRE
  For any $E\in \mathcal{E}$ s.t. $|E|=2$, $|\funcdef{CommonR}{E}|\le1$;
  and
  for any $h\in H$, $\hq{h}\le2$.
  \item[\textbf{Output:}] A strongly stable matching{} $M$ of $I$.
  \STATE Apply RGS{} algorithm to $I$ and let $M$ be the obtained stable matching.
  \label{L46}
  \WHILE{there exists $E \in \mathcal{E}$ such that $|M(E)| > \hc{E}$%
            } \label{L47}
       \IF{$|E| = 1$} \label{L48}
         \STATE Let $h'$ be the only hospital in $E$. 
       \ELSIF[$|E|=2$ also holds]{$|\funcdef{CommonR}{E}| = 1$} \label{L49}
         \STATE Let $r$ be the only resident in $\funcdef{CommonR}{E}$.
         \STATE Let $h_+$ and $h_-$ be the two hospitals in $E$ such that $h_+ \succ_{r} h_-$.
         \STATE Let $h'$ be $h_-$ if $\hq{h_-} > 0$, and $h_+$ otherwise. \label{L50}
       \ELSE[i.e., when $|E|=2$ and $|\funcdef{CommonR}{E}|=0$] \label{L51}
         \STATE Let $h'$ be any hospital in $E$ such that $\hq{h'} > 0$. 
       \ENDIF
       \STATE Modify $I$ by updating $\hq{h'} := \hq{h'} - 1$. \label{L52}
       \STATE Apply RGS{} algorithm to $I$ and let $M$ be the obtained stable matching.
       \label{L53}
  \ENDWHILE
  \RETURN $M$ \label{L54}
 \end{calgorithmic}
\end{algorithm}
\begin{lemma}%
 \label{L55}

\newcommand{\hq}[1]{q\hifempty{#1}{}{(#1)}} %

\newcommand{\hc}[1]{c\hifempty{#1}{}{(#1)}} %

\newcommand{\hpreflist}[2]{
\begin{tabular}{rccccccccccccccccccccccc}
 #1
\end{tabular}
\hspace{1cm}
\begin{tabular}{rccccccccccccccccccccccc}
 #2
\end{tabular}
}
 Given an \textsc{HRRC}{}-instance $I = (R, H, \succ_{}, \hq{}, \mathcal{E}, \hc{})$
 as described above,
 \cref{L40} finds a
 strongly stable matching{} in $O( |H| (|H| + |R|) )$ time.
\end{lemma}
\begin{proof}

\newcommand{\hq}[1]{q\hifempty{#1}{}{(#1)}} %

\newcommand{\hc}[1]{c\hifempty{#1}{}{(#1)}} %

\newcommand{\hpreflist}[2]{
\begin{tabular}{rccccccccccccccccccccccc}
 #1
\end{tabular}
\hspace{1cm}
\begin{tabular}{rccccccccccccccccccccccc}
 #2
\end{tabular}
}

Let $I_0$ be $I$ at the beginning of the algorithm, and
 $I_1$ and $M_1$ be $I$ and $M$ at the termination of the algorithm, respectively.
 We show that $M_1$
 is a strongly stable matching{} of $I_0$.
 We first show the feasibility.
 By the condition of the \textbf{while}-loop at \cref{L47}, $M_1$ is a feasible{}{} matching of $I_1$.
 Note that the regional caps are same in $I_0$ and $I_1$,
 and each hospital's capacity in $I_{1}$ is at most that in $I_{0}$.
 Thus, $M_1$ is a feasible{}{} matching of $I_0$.

 We then show the stability.
 Suppose that $M_1$ is not strongly stable in $I_0$ and
 let $(r_*, h_*)$ be an SBP{}{}.
 We show that
 there exists a region $E_*$ such that $h_*\in E_*$ and $|M_1(E_*)| = \hc{E_*}$.
 Note that $(r_*,h_*)$ is a BP{}{} for $M_1$ in $I_0$ but not in $I_1$.
 Since the only difference between $I_0$ and $I_1$
 are the capacities of the hospitals,
 $h_*$'s capacity must be reduced at \cref{L52} during the course of the algorithm.
 Let $E_*$ be the region selected in this round, i.e., $h_*\in E_*$.
 Note that $E_*$ is uniquely determined since $\mathcal{E}$ is disjoint.
 Let $t$ be the last time when $h_*$'s capacity is reduced.
 Just before $t$, $|M(E_*)| > \hc{E_*}$ holds and just after $t$, $|M(E_*)| \ge \hc{E_*}$ holds by \cref{L45}(i).
 Also, by \cref{L45}(ii), $|M(E_*)| \ge \hc{E_*}$ holds till the end of the algorithm.
 Therefore, $|M_1(E_*)| \ge \hc{E_*}$ holds.
 On the other hand, since $M_1$ is feasible{}{} in $I_1$, we have that $|M_1(E_*)| \le \hc{E_*}$.
 Thus $|M_1(E_*)|=\hc{E_*}$.

 Next, we show that $M_1(r_*) \ne \emptyset$ and $E_* = \set{ h_*, M_1(r_*) }$.
Since $(r_*,h_*)$ is a BP{}{} for $M_1$ in $I_{0}$, 
 $r_*$ is unassigned or $h_* \succ_{r_*} M_1(r_*)$.
 But since $(r_*,h_*)$ is not a BP{}{} for $M_1$ in $I_{1}$, 
 for any $r\in M_1(h_*)$, $r\succ_{h_*}r_*$ holds.
 As $(r_*,h_*)$ is an SBP{}{} for $M_1$ in $I_0$,
 by \cref{L13},
 $M_1 \setminus \set{(r_*,M_1(r_*))}\cup \set{(r_*,h_*)}$ is feasible{}{} in $I_0$.
 This implies that $r_*$ is currently assigned to a hospital in $E_*$
 because $|M_1(E_*)|=\hc{E_*}$ but moving $r_*$ to $h_*$ does not violate
 the regional cap of $E_*$.
 That is, 
 we have $M_1(r_*)\ne \emptyset$ and $M_1(r_*) \in E_*$.
 Since $h_* \ne M_1(r_*)$ and $|E_*|\le2$, 
 we have that $E_* = \set{h_*, M_1(r_*)}$.

 Since $r_*$ is acceptable to both $h_*$ and $M_1(r_*)$,
 we have that $r_* \in \funcdef{CommonR}{E_*}$.
 But since $|\funcdef{CommonR}{E_*}|\le1$ by assumption,
 we have that $|\funcdef{CommonR}{E_*}|=1$.
 As mentioned above, $h_*$'s capacity was reduced,
 and when it happened $E_*$ was treated at \cref{L49}
 since $|E_*|=2$ and $|\funcdef{CommonR}{E_*}|=1$.
 Since $h_* \succ_{r_*} M_1(r_*)$ and $h_*$'s capacity was reduced,
 the condition of \cref{L50} implies that $\hq{ M_1(r_*) } = 0$ holds in $I_1$.
 This contradicts that 
 $M_1$ is a matching of $I_1$.
 Thus, $M_1$ is a strongly stable matching{} of $I_0$.

Finally, we show the time complexity.
 Using suitable data structures,
 RGS{} algorithm can be implemented to run in $O(m+|H|+|R|)$ time \cite{GI89},
 where $m$ is the number of acceptable pairs in an input.
 Since $|\funcdef{AC}{h}|\le2$ for any $h\in H$,
 $m \le 2|H|$ holds, and hence each of
 \cref{L46,L53} can be executed in $O(|H|+|R|)$ time.
 Each step in the \textbf{while}-loop, except for \cref{L53}, can be executed in constant time.
 \cref{L47} can be executed in $O(|H|)$ time since $|\mathcal{E}| \le |H|$.
 The number of iterations of the \textbf{while}-loop is
 at most $2|H|$,
 since
 the capacity of one hospital decreases by one per iteration
 and the total capacity of the hospitals is at most $2|H|$, thanks to the shrinking operation.
 Thus, \cref{L40} reaches
 \cref{L54} in $O( |H| (|H| + |R|) )$ time
 and outputs $M$.
\end{proof}

\subsubsection{Main Algorithm}
\label{L42}
Finally,
we construct the main algorithm using the previous sub-algorithm.

\begin{theorem}
 \label{L7}

\newcommand{\hq}[1]{q\hifempty{#1}{}{(#1)}} %

\newcommand{\hc}[1]{c\hifempty{#1}{}{(#1)}} %

\newcommand{\hpreflist}[2]{
\begin{tabular}{rccccccccccccccccccccccc}
 #1
\end{tabular}
\hspace{1cm}
\begin{tabular}{rccccccccccccccccccccccc}
 #2
\end{tabular}
}
 There exists an $O(|H| (|H|+|R|) )$-time algorithm to find a
 strongly stable matching{} or report that none exists,
 given a $(2,2,2)$-\textsc{HRRCDR}{}-instance $I = (R, H, \succ_{}, \hq{}, \mathcal{E}, \hc{})$.
\end{theorem}
\begin{proof}

\newcommand{\hq}[1]{q\hifempty{#1}{}{(#1)}} %

\newcommand{\hc}[1]{c\hifempty{#1}{}{(#1)}} %

\newcommand{\hpreflist}[2]{
\begin{tabular}{rccccccccccccccccccccccc}
 #1
\end{tabular}
\hspace{1cm}
\begin{tabular}{rccccccccccccccccccccccc}
 #2
\end{tabular}
}

\cref{L56}
 achieves our goal.
First suppose that \cref{L56} returns a matching $M$ at \cref{L57}.
Then each $M_{i}$ ($i\in[1,p]$) is a strongly stable matching for $I_i$.
 Also \cref{L55,L44}
 imply that $M_0$ is a strongly stable matching{} for $I_0$.
Since $I_{0}, I_{1}, \ldots, I_{p}$ are disjoint,
 there is no acceptable pair between different subinstances.
Also, there is no SBP within $M_{i}$, so $M$ is a strongly stable matching for $I$.
Next, suppose that \cref{L56} returns $\bot$ at \cref{L58}.
Then $M_i = \bot$ for some $i \in [1,p]$.
This implies that any feasible matching for $I$ has an SBP{}{} within $I_{i}$
and hence $I$ admits no strongly stable matching{}.
\begin{algorithm}[tb]
 \caption{An algorithm for \textsc{Strong-$(2,2,2)$-HRRCDR}.}
 \label{L56}
 \begin{calgorithmic}[1]
  \item[\textbf{Input:}] A $(2,2,2)$-\textsc{HRRCDR}{}-instance $I = (R, H, \succ_{}, \hq{}, \mathcal{E}, \hc{})$.
  \item[\textbf{Output:}] A strongly stable matching{} $M$ of $I$ if it exists, $\bot$ otherwise. %
  \STATE Find all the $(2 \times 2)$-subinstances in $I$ and let them be $I_i := (R_i, H_i, \succ_{}, \hq{}, \mathcal{E}_i, \hc{})$ for $i \in [1,p]$, where $\mathcal{E}_i := \set{ H_i }$.
  \label{L59}
  \STATE Remove agents and regions of all $I_1,\dots,I_p$ from $I$, and let 
         $I_0 := (R_0, H_0, \succ_{}, \hq{}, \mathcal{E}_0, \hc{})$ be the resultant instance.
  \label{L60}
  \STATE For each $i \in [1,p]$,
  find a strongly stable matching{} of $I_i$ by a brute-force search.
  Let $M_i$ be an obtained strongly stable matching{} if found; $M_i := \bot$ otherwise.
         \label{L61}
  \STATE Apply \cref{L40} to $\funcdef{Shrink}{I_0}$ and obtain an output $M_0$.
         \label{L62}
  \IF{for any $i \in [0,p]$, $M_i \ne \bot$}
  \label{L63}
  \RETURN $M=\bigcup_{i\in[0,p]} M_i$
  \label{L57}
  \ELSE
  \RETURN $\bot$
  \label{L58}
  \ENDIF
 \end{calgorithmic}
 \end{algorithm}

  We then show the time complexity.
 Since each region's size is at most two,
 \cref{L59} can be executed in $O(|\mathcal{E}|)$ time.
 Construction of $I_0$ in \cref{L60}
 can be done in $O(|H|+|R|)$ time.
 Since each $I_i$ ($i \in [1,p]$) is of constant size,
 each brute-force search in \cref{L61} runs in $O(1)$ time.
 Invocation of $\funcdef{Shrink}{}$ in \cref{L62}
 runs in $O(|H|)$ time.
 By \cref{L55},
 \cref{L40} can be executed
 in $O(|H| (|H| + |R|) )$ time.
 \cref{L63,L57}
 can be executed in $O(p)$ time and $O(|H|)$ time, respectively.
 Since 
 the number of $(2 \times 2)$-subinstances, $p$, is at most $|\mathcal{E}|$
 and $|\mathcal{E}| \le |H|$,
 the whole time complexity is $O(|H| (|H|+|R|) )$.
\end{proof}

  \subsection{\textsc{PPN-3-SAT}{}  }
  \label{L35}
To make our reductions for hardness results simpler,
we introduce a restricted variant of \textsc{SAT},
which we call \textsc{PPN-3-SAT}{},
and show its NP-completeness.
An instance of \textsc{PPN-3-SAT}{} consists of
a set $U$ of variables and a set $C$ of clauses over $U$
such that
for any $c\in C$, $c$ has two or three literals, and
for any $x\in U$, $x$ appears exactly twice
and $\bar{x}$ appears exactly once in $C$.
\textsc{PPN-3-SAT}{} asks if
there exists a satisfying assignment.
We will prove that \textsc{PPN-3-SAT}{} is NP-complete.
The proof is almost the same as
the proof of Theorem 2.1 in \cite{Tov84}.
\begin{lemma}
 \label{L64}
 \textsc{PPN-3-SAT}{} is NP-complete.
\end{lemma}
  \begin{proof}
{

\newcommand{\hq}[1]{q\hifempty{#1}{}{(#1)}} %

\newcommand{\hc}[1]{c\hifempty{#1}{}{(#1)}} %

\newcommand{\hpreflist}[2]{
\begin{tabular}{rccccccccccccccccccccccc}
 #1
\end{tabular}
\hspace{1cm}
\begin{tabular}{rccccccccccccccccccccccc}
 #2
\end{tabular}
}

 Membership in NP is obvious.
 We will reduce from \textsc{3-SAT}{}.
 Given a \textsc{3-SAT}{}-instance,
 we construct a \textsc{PPN-3-SAT}{}-instance as follows.
 For each variable $x$
 in the original \textsc{3-SAT}{} instance,
 perform the following procedure:
 suppose $x$ appears $k$ ($k \ge 1$) times.
 Create $k$ new variables $x_1, \dots, x_k$ and
 for $i = 1, \dots, k$,
 replace
 the $i$th occurrence of $x$
 with $x_i$.
 Append the clause $\{ \ell_i \vee \overline{\ell_{i+1}} \}$ for
 $i = 1, \dots, k-1$ and the clause $\{ \ell_k \vee \overline{\ell_1} \}$,
 where
 $\ell_i = x_{i}$
 if the $i$th occurrence of $x$
 is positive and
 $\ell_i = \overline{x_{i}}$
 otherwise.

 In the new instance, the clause $\{ \ell_i \vee \overline{\ell_{i+1}} \}$ implies that
 if $\ell_i$ is false, $\ell_{i+1}$ must be false as well.
 The cyclic structure of the clauses therefore forces $\ell_i$ to be
 either all true or all false, so the new instance is satisfiable if and only if the
 original one is.
 Moreover the transformation requires polynomial time.
}
  \end{proof}

  \subsection{NP-completeness of \textsc{Strong-$(2,2,3)$-HRRCDR} }
  \label{L36}
\begin{theorem}%
 \label{L8}

 \textsc{Strong-$(2,2,3)$-HRRCDR} is NP-complete.
\end{theorem}
\begin{proof}

\newcommand{\hq}[1]{q\hifempty{#1}{}{(#1)}} %

\newcommand{\hc}[1]{c\hifempty{#1}{}{(#1)}} %

\newcommand{\hpreflist}[2]{
\begin{tabular}{rccccccccccccccccccccccc}
 #1
\end{tabular}
\hspace{1cm}
\begin{tabular}{rccccccccccccccccccccccc}
 #2
\end{tabular}
}

 Membership in NP is obvious.
 We show a reduction from an NP-complete problem \textsc{PPN-3-SAT}{}.
 Let $I$ be an instance of \textsc{PPN-3-SAT}{} having
 $n$ variables $x_{i}$ ($i \in [1,n]$) and
 $m$ clauses $C_{j}$ ($j \in [1,m]$).
 For $k = 2, 3$, we call a clause containing $k$ literals a
 {\em $k$-clause}.
 Suppose that there are
 $m_2$ $2$-clauses
 and
 $m_3$ $3$-clauses
 (thus $m_2 + m_3 = m$),
 and assume without loss of generality that
 $C_{j}$ ($j \in [1,m_2]$) are $2$-clauses
 and
 $C_{j}$ ($j \in [m_2+1,m]$) are $3$-clauses.
 For each variable $x_{i}$, we construct a
 {\em variable gadget{}}.
 It consists of
 two residents $e'_{i,1}$ and $e'_{i,2}$,
 seven hospitals
 $b'_{i,1}$, $x'_{i,1}$,
 $b'_{i,2}$, $x'_{i,2}$,
 $b'_{i,3}$, $b'_{i,4}$, and $x'_{i,3}$,
 and
 three regions
 $\set{ b'_{i,1}, x'_{i,1} }$,
 $\set{ b'_{i,2}, x'_{i,2} }$, and
 $\set{ b'_{i,3}, b'_{i,4}, x'_{i,3} }$.
 A variable gadget{} corresponding to $x_{i}$ is called an {\em $x_{i}$-gadget}.
 For each clause $C_{j}$,
 we construct a {\em clause gadget{}}.
 If $C_{j}$ is a $2$-clause, we create
 two residents $c'_{j,1}$ and $c'_{j,2}$,
 two hospitals $a'_{j,1}$ and $y'_{j}$,
 and
 a region $\set{a'_{j,1}, y'_{j}}$.
 If $C_{j}$ is a $3$-clause, we create
 four residents $c'_{j,1}$, $c'_{j,2}$, $d'_{j}$, and $c'_{j,3}$,
 four hospitals $a'_{j,1}$, $a'_{j,2}$, $a'_{j,3}$, and $y'_{j}$,
 and
 two regions $\set{a'_{j,1}, a'_{j,2}}$ and $\set{a'_{j,3}, y'_{j}}$.
 A clause gadget{} corresponding to $C_{j}$ is called a {\em $C_{j}$-gadget}.
 For each clause $C_{j}$,
 we also construct a {\em terminal gadget{}}.
 We create
 three residents $g'_{j,1}$, $g'_{j,3}$, and $z'_{j}$,
 three hospitals $g'_{j,2}$, $g'_{j,4}$, and $t'_{j}$,
 and
 a region $\set{g'_{j,2}, g'_{j,4}, t'_{j}}$.
 A terminal gadget{} corresponding to $C_{j}$ is called a {\em $T_{j}$-gadget}.
 Thus, there are
 $2n + 2m_2 + 4m_3 + 3m$ residents,
 $7n + 2m_2 + 4m_3 + 3m$ hospitals, and
 $3n +  m_2 + 2m_3 +  m$ regions
 in the created \textsc{HRRCDR}{} instance, denoted $I'$.

 \begin{figure}[tb]
  \centering
  \hpreflist{
  $e'_{i,1}$: & $b'_{i,1}$ & $b'_{i,3}$ \\
  \\
  \\
  $e'_{i,2}$: & $b'_{i,2}$ & $b'_{i,4}$ \\
  \\
  \\
  \\
  \\
  \\
  }{
  $b'_{i,1}[1]$: & $e'_{i,1}$ \\
  $x'_{i,1}[1]$: & $\underline{ c'_{j_{i,1},\ell_{i,1}} }$ \\ %
  \\
  $b'_{i,2}[1]$: & $e'_{i,2}$ \\
  $x'_{i,2}[1]$: & $\underline{ c'_{j_{i,2},\ell_{i,2}} }$ \\ %
  \\
  $b'_{i,3}[1]$: & $e'_{i,1}$ \\
  $b'_{i,4}[1]$: & $e'_{i,2}$ \\
  $x'_{i,3}[1]$: & $\underline{ c'_{j_{i,3},\ell_{i,3}} }$ \\ %
  } \\
  \begin{tabular}{c}
   $\hc{ \set{ b'_{i,1}, x'_{i,1} } } := 1$,
   $\hc{ \set{ b'_{i,2}, x'_{i,2} } } := 1$,
   $\hc{ \set{ b'_{i,3}, b'_{i,4}, x'_{i,3} } } := 2$
  \end{tabular}
  \caption{Preference lists, capacities, and regional caps of $x_{i}$-gadget. Here and hereafter, for readability, agents defined outside the gadget are underlined. }
  \label{L65}
 \end{figure}
 Let us construct preference lists of variable gadgets.
 Suppose that $x_{i}$'s $k$th positive occurrence ($k = 1, 2$) is in the
 $j_{i,k}$th clause $C_{j_{i,k}}$ as
 the $\ell_{i,k}$th literal ($1 \le \ell_{i,k} \le 3$).
 Similarly, suppose that $x_{i}$'s negative occurrence is in the
 $j_{i,3}$th clause $C_{j_{i,3}}$ as the $\ell_{i,3}$th literal ($1 \le \ell_{i,3} \le 3$).
 Then, preference lists, capacities, and regional caps of the $x_{i}$-gadget are
 constructed as shown in
 \cref{L65}.%
 \begin{figure}[tb]
  \centering
  \hpreflist{
  $c'_{j,1}$: & $\underline{ x'_{i_{j,1},k_{j,1}} }$ & $a'_{j,1}$ \\
  $c'_{j,2}$: & $\underline{ x'_{i_{j,2},k_{j,2}} }$ & $a'_{j,1}$ \\
  }{
  $a'_{j,1}[1]$: & $c'_{j,1}$ & $c'_{j,2}$ \\
  $y'_{j}[1]$: & $\underline{ z'_{j} }$ \\
  } \\
  \begin{tabular}{c}
   $\hc{ \set{a'_{j,1}, y'_{j}} } := 1$
  \end{tabular}
  \caption{Preference lists, capacities, and regional caps of $C_{j}$-gadget ($j \in [1,m_2]$). }
  \label{L66}
 \end{figure}
 \begin{figure}[tb]
  \centering
  \hpreflist{
  $c'_{j,1}$: & $\underline{ x'_{i_{j,1},k_{j,1}} }$ & $a'_{j,1}$ \\
  $c'_{j,2}$: & $\underline{ x'_{i_{j,2},k_{j,2}} }$ & $a'_{j,1}$ \\
  \\
  $d'_{j}$: & $a'_{j,2}$ & $a'_{j,3}$ \\
  $c'_{j,3}$: & $\underline{ x'_{i_{j,3},k_{j,3}} }$ & $a'_{j,3}$ \\
  }{
  $a'_{j,1}[1]$: & $c'_{j,1}$ & $c'_{j,2}$ \\
  $a'_{j,2}[1]$: & $d'_{j}$ \\
  \\
  $a'_{j,3}[1]$: & $d'_{j}$ & $c'_{j,3}$ \\
  $y'_{j}[1]$: & $\underline{ z'_{j} }$ \\
  } \\
  \begin{tabular}{c}
  $\hc{ \set{a'_{j,1}, a'_{j,2}} } := 1$,
  $\hc{ \set{a'_{j,3}, y'_{j}} } := 1$
  \end{tabular}
  \caption{Preference lists, capacities, and regional caps of $C_{j}$-gadget ($j \in [m_2+1,m]$). }
  \label{L67}
 \end{figure}
 We then construct preference lists of clause gadget{}s.
 Consider a clause $C_{j}$, and suppose that its $\ell$th literal is of a
 variable $x_{i_{j,\ell}}$.
 Define $k_{j,\ell}$ as
 \[
 k_{j,\ell} := \pcase{ll}{
 1 & \text{if this is the first positive occurrence of $x_{i_{j,\ell}}$,} \\
 2 & \text{if this is the second positive occurrence of $x_{i_{j,\ell}}$, and} \\
 3 & \text{if this is the negative occurrence of $x_{i_{j,\ell}}$.}
 }
 \]
 If $C_{j}$ is a $2$-clause (respectively, $3$-clause),
 then the preference lists, capacities, and regional caps of the $C_{j}$-gadget are shown in
 \cref{L66} (respectively, \cref{L67}).%
\begin{figure}[tb]
  \centering
  \hpreflist{
  $g'_{j,1}$: & $g'_{j,2}$ & $g'_{j,4}$ \\
  $g'_{j,3}$: & $g'_{j,4}$ & $g'_{j,2}$ \\
  $z'_{j}$: & $\underline{ y'_{j} }$ & $t'_{j}$ \\
  }{
  $g'_{j,2}[1]$: & $g'_{j,3}$ & $g'_{j,1}$ \\
  $g'_{j,4}[1]$: & $g'_{j,1}$ & $g'_{j,3}$ \\
  $t'_{j}[1]$: & $z'_{j}$ \\
  } \\
  \begin{tabular}{c}
  $\hc{ \set{g'_{j,2}, g'_{j,4}, t'_{j}} } := 1$
  \end{tabular}
  \caption{Preference lists, capacities, and regional caps of $T_{j}$-gadget. }
  \label{L68}
\end{figure}%
Finally, we construct preference lists of terminal gadget{}s.
 The preference lists, capacities, and regional caps of
 the $T_{j}$-gadget are shown in \cref{L68}.
 Now the reduction is completed.
 It is not hard to see that the reduction
 can be performed in polynomial time and
 $I'$ is an instance of $(2,2,3)$-\textsc{HRRCDR}{}.

 It might be helpful to roughly explain intuition behind the construction of $I'$.
 Consider the $x_{i}$-gadget in \cref{L65} and let $E_{i,1}=\set{ b'_{i,1}, x'_{i,1} }$,
 $E_{i,2}=\set{ b'_{i,2}, x'_{i,2} }$, and $E_{i,3}=\set{ b'_{i,3}, b'_{i,4}, x'_{i,3} }$ be the regions.
 Three hospitals $x'_{i,1}$, $x'_{i,2}$, and $x'_{i,3}$ respectively correspond to the first positive occurrence,
 the second positive occurrence, and the negative occurrence of $x_{i}$.
 There are two ``proper'' assignments for an $x_{i}$-gadget,
 $M_{i,0}=\set{ (e'_{i,1}, b'_{i,3}), (e'_{i,2}, b'_{i,4}) }$ and $M_{i,1}=\set{ (e'_{i,1}, b'_{i,1}), (e'_{i,2}, b'_{i,2}) }$.
 We associate an assignment $x_{i}=0 $ with $M_{i,0}$ and $x_{i}=1$ with $M_{i,1}$.
 If $M_{i,0}$ is chosen, the region $E_{i,3}$ becomes full{}, so the hospital $x'_{i,3}$,
 corresponding to $\overline{x_{i}}$, is ``safe'' in the sense that $x'_{i,3}$ will not create an SBP{}{} even if it is left empty.
 Similarly, if $M_{i,1}$ is chosen, the regions $E_{i,1}$ and $E_{i,2}$ become full{},
 so both hospitals $x'_{i,1}$ and $x'_{i,2}$, corresponding to $x_{i}$, are safe.
 Note that a matching other than these two are not beneficial.
 For example, adopting a matching $\set{ (e'_{i,1}, b'_{i,1}), (e'_{i,2}, b'_{i,4}) }$ leaves $E_{i,2}$ and $E_{i,3}$ deficient{},
 so only $x'_{i,1}$ is safe, which is strictly worse than choosing $M_{i,1}$.

 Next, consider a $C_{j}$-gadget for a $2$-clause $C_{j}$ and see \cref{L66}.
 The residents $c'_{j,1}$ and $c'_{j,2}$ correspond to the first and the second literals of $C_{j}$, respectively.
 If the $\ell$th literal ($\ell=1, 2$) takes the value 0 by an assignment,
 then $c'_{j,\ell}$ must be assigned to the first-choice hospital (from a variable gadget)
 to avoid forming an SBP{}{} because such a hospital is ``unsafe'' as we saw above.
 Therefore, if $C_{j}$ is unsatisfied by an assignment, the hospital $a'_{j,1}$ must remain empty.
 Conversely, if $C_{j}$ is satisfied, $a'_{j,1}$ can be assigned a resident corresponding to a literal taking the value 1.
 The argument for a $3$-clause is a bit more complicated, but the same conclusion holds:
 $C_{j}$ is satisfied if and only if we can construct an assignment such that $a'_{j,3}$ is assigned a resident.
 Now let $E_{j}$ be the region $\set{a'_{j,1}, y'_{j} \}$ if $C_{j}$ is a $2$-clause and $\{a'_{j,3}, y'_{j} }$ if $C_{j}$ is a $3$-clause.
 Then $E_{j}$ is deficient{} if and only if $C_{j}$ is unsatisfied.

 Note that the resident $z'_{j}$ in the $T_{j}$-gadget and the hospital $y'_{j}$ in the $C_{j}$-gadget are mutually first-choice.
 Therefore, if $E_{j}$ is deficient{}, we have to assign $z'_{j}$ to $y'_{j}$, as otherwise they form an SBP{}{}.
 In this case, $T_{j}$-gadget reduces to the instance $G_{2}$ in \cref{L14}, so at least one SBP{}{} yields from this gadget.
 On the other hand, if $E_{j}$ is full, we can assign $z'_{j}$ to $t'_{j}$, so the $T_{j}$-gadget will not create an SBP{}{}.
 In this way, we can associate a satisfying assignment of $I$ with a strongly stable matching of $I'$. 

 Now we start a formal correctness proof.
 First, suppose that $I$ is satisfiable and
 let $A$ be a satisfying assignment.
 We construct a strongly stable matching{} $M$ of
 $I'$ from $A$ as follows.
 We say that a $2$-clause $C$ is
 {\em $(v_1,v_2)$-satisfied} by $A$ ($v_i\in \set{0,1}$)
 if the first and the second literals of $C$ take the values $v_1$ and $v_2$,
 respectively, in $A$.
 A $3$-clause is defined to be {\em $(v_1,v_2,v_3)$-satisfied} analogously.
 For an $x_{i}$-gadget,
 we consider the following two cases.
 \begin{itemize}
  \item If $x_{i} = 0$ by $A$, then add $(e'_{i,1}, b'_{i,3})$ and $(e'_{i,2}, b'_{i,4})$ to $M$.
  \item If $x_{i} = 1$ by $A$, then add $(e'_{i,1}, b'_{i,1})$ and $(e'_{i,2}, b'_{i,2})$ to $M$.
 \end{itemize}
 For a $C_{j}$-gadget ($j \in [1,m_2]$), we consider the following three cases.
 \begin{itemize}

  \item If $C_{j}$ is $(0,1)$-satisfied by $A$, then add $(c'_{j,1}, x'_{i_{j,1},k_{j,1}})$ and $(c'_{j,2},a'_{j,1})$ to $M$.
  \item If $C_{j}$ is $(1,0)$-satisfied by $A$, then add $(c'_{j,1},a'_{j,1})$ and $(c'_{j,2}, x'_{i_{j,2},k_{j,2}})$ to $M$.
  \item If $C_{j}$ is $(1,1)$-satisfied by $A$, then add $(c'_{j,1},a'_{j,1})$ to $M$.
 \end{itemize}
 For a $C_{j}$-gadget ($j \in [m_2+1,m]$), we consider the following seven cases.
 \begin{itemize}

  \item If $C_{j}$ is $(0,0,1)$-satisfied by $A$, then add $(c'_{j,1}, x'_{i_{j,1},k_{j,1}})$, $(c'_{j,2}, x'_{i_{j,2},k_{j,2}})$, $(d'_{j}, a'_{j,2})$, and $(c'_{j,3}, a'_{j,3})$ to $M$.
  \item If $C_{j}$ is $(0,1,0)$-satisfied by $A$, then add $(c'_{j,1}, x'_{i_{j,1},k_{j,1}})$, $(c'_{j,2}, a'_{j,1})$, $(d'_{j}, a'_{j,3})$, and $(c'_{j,3}, x'_{i_{j,3},k_{j,3}})$ to $M$.
  \item If $C_{j}$ is $(1,0,0)$-satisfied by $A$, then add $(c'_{j,1}, a'_{j,1})$, $(c'_{j,2}, x'_{i_{j,2},k_{j,2}})$, $(d'_{j}, a'_{j,3})$, and $(c'_{j,3}, x'_{i_{j,3},k_{j,3}})$ to $M$.
  \item If $C_{j}$ is $(0,1,1)$-satisfied by $A$, then add $(c'_{j,1}, x'_{i_{j,1},k_{j,1}})$, $(c'_{j,2}, a'_{j,1})$, and $(d'_{j}, a'_{j,3})$ to $M$.
  \item If $C_{j}$ is $(1,0,1)$-satisfied by $A$, then add $(c'_{j,1}, a'_{j,1})$, $(c'_{j,2}, x'_{i_{j,2},k_{j,2}})$, and $(d'_{j}, a'_{j,3})$ to $M$.
  \item If $C_{j}$ is $(1,1,0)$-satisfied by $A$, then add $(c'_{j,1}, a'_{j,1})$, $(d'_{j}, a'_{j,3})$, and $(c'_{j,3}, x'_{i_{j,3},k_{j,3}})$ to $M$.
  \item If $C_{j}$ is $(1,1,1)$-satisfied by $A$, then add $(c'_{j,1}, a'_{j,1})$ and $(d'_{j}, a'_{j,3})$ to $M$.
 \end{itemize}
 For a $T_{j}$-gadget, add $(z'_{j},t'_{j})$ to $M$.
 It is easy to see that $M$ satisfies all the capacity constraints
 of hospitals and regions, and hence is a feasible{}{} matching.

 Next, we prove that $M$ is strongly stable in $I'$
 by showing that no hospital in $I'$ can be a part of an SBP{}{}.
For hospitals in $x_{i}$-gadget,
 we consider two cases.
 \begin{description}
  \item[{\boldmath $x_{i} = 0$.}]
	     For each $h \in \set{b'_{i,1},b'_{i,2},x'_{i,3}}$, $h$ is empty, but the region $h$ belongs to is full and no acceptable resident to $h$ is assigned to a hospital in the same region. %
	     For each $h \in \set{x'_{i,1},x'_{i,2},b'_{i,3},b'_{i,4}}$, $h$ is assigned the first-choice resident. %
  \item[{\boldmath $x_{i} = 1$}.]
	     For each $h \in \set{b'_{i,1},b'_{i,2},x'_{i,3}}$, $h$ is assigned the first-choice resident. %
	     For each $h \in \set{x'_{i,1},x'_{i,2}}$, $h$ is empty, but the region $h$ belongs to is full and no acceptable resident to $h$ is assigned to a hospital in the same region. %
	     For each $h \in \set{b'_{i,3},b'_{i,4}}$, $h$ is empty, but each acceptable resident to $h$ is assigned to a better hospital than $h$. %
 \end{description}
For hospitals in $C_{j}$-gadget ($j \in [1,m_2]$), we consider three cases.
       \begin{description}

\item[{\boldmath $C_{j}$ is $(0,1)$-satisfied.}]
					    Hospital $a'_{j,1}$ is assigned the second-choice resident $c'_{j,2}$, and the first-choice resident $c'_{j,1}$ is assigned to a better hospital $x'_{i_{j,1},k_{j,1}}$ than $a'_{j,1}$. %
					    Hospital $y'_{j}$ is empty, the region $y'_{j}$ belongs to is full, and $z'_{j}$ is not assigned to a hospital in the same region. %
\item[{\boldmath $C_{j}$ is $(1,0)$-satisfied.}]
					    Hospital $a'_{j,1}$ is assigned the first-choice resident $c'_{j,1}$. %
					    Hospital $y'_{j}$ is empty, the region $y'_{j}$ belongs to is full, and $z'_{j}$ is not assigned to a hospital in the same region. %
\item[{\boldmath $C_{j}$ is $(1,1)$-satisfied.}]
					    Hospital $a'_{j,1}$ is assigned the first-choice resident $c'_{j,1}$. %
					    Hospital $y'_{j}$ is empty, the region $y'_{j}$ belongs to is full, and $z'_{j}$ is not assigned to a hospital in the same region. %
       \end{description}
For hospitals in $C_{j}$-gadget ($j \in [m_2+1,m]$), we consider seven cases.
 \begin{description}

\item[{\boldmath $C_{j}$ is $(0,0,1)$-satisfied.}]
				      For each $h \in \set{a'_{j,1},y'_{j}}$, $h$ is empty, but the region $h$ belongs to is full and no acceptable resident to $h$ is assigned to a hospital in the same region. %
				      Hospital $a'_{j,2}$ is assigned the first-choice resident $d'_{j}$. %
				      Hospital $ a'_{j,3} $ is assigned the second-choice resident $ c'_{j,3} $, and the first-choice resident $ d'_{j} $ is assigned to a better hospital $ a'_{j,2} $ than $ a'_{j,3} $. %
\item[{\boldmath $C_{j}$ is $(0,1,0)$-satisfied.}]
				      Hospital $ a'_{j,1} $ is assigned the second-choice resident $ c'_{j,2} $, and the first-choice resident $ c'_{j,1} $ is assigned to a better hospital $ x'_{i_{j,1},k_{j,1}} $ than $ a'_{j,1} $. %
				      For each $h \in \set{a'_{j,2},y'_{j}}$, $h$ is empty, but the region $h$ belongs to is full and no acceptable resident to $h$ is assigned to a hospital in the same region. %
				      Hospital $a'_{j,3}$ is assigned the first-choice resident $d'_{j}$. %

\item[{\boldmath $C_{j}$ is $(1,0,0)$-satisfied.}]
				      For each $h \in \set{a'_{j,1},a'_{j,3}}$, $h$ is assigned the first-choice resident. %
				      For each $h \in \set{a'_{j,2},y'_{j}}$, $h$ is empty, but the region $h$ belongs to is full and no acceptable resident to $h$ is assigned to a hospital in the same region. %
\item[{\boldmath $C_{j}$ is $(0,1,1)$-satisfied.}]
				      Hospital $ a'_{j,1} $ is assigned the second-choice resident $ c'_{j,2} $, and the first-choice resident $ c'_{j,1} $ is assigned to a better hospital $ x'_{i_{j,1},k_{j,1}} $ than $ a'_{j,1} $. %
				      For each $h \in \set{a'_{j,2},y'_{j}}$, $h$ is empty, but the region $h$ belongs to is full and no acceptable resident to $h$ is assigned to a hospital in the same region. %
				      Hospital $a'_{j,3}$ is assigned the first-choice resident $d'_{j}$. %

\item[{\boldmath $C_{j}$ is $(1,0,1)$-satisfied.}]
				      For each $h \in \set{a'_{j,1},a'_{j,3}}$, $h$ is assigned the first-choice resident. %
				      For each $h \in \set{a'_{j,2},y'_{j}}$, $h$ is empty, but the region $h$ belongs to is full and no acceptable resident to $h$ is assigned to a hospital in the same region. %
\item[{\boldmath $C_{j}$ is $(1,1,0)$-satisfied.}]
				      For each $h \in \set{a'_{j,1},a'_{j,3}}$, $h$ is assigned the first-choice resident. %
				      For each $h \in \set{a'_{j,2},y'_{j}}$, $h$ is empty, but the region $h$ belongs to is full and no acceptable resident to $h$ is assigned to a hospital in the same region. %
\item[{\boldmath $C_{j}$ is $(1,1,1)$-satisfied.}]
				      For each $h \in \set{a'_{j,1},a'_{j,3}}$, $h$ is assigned the first-choice resident. %
				      For each $h \in \set{a'_{j,2},y'_{j}}$, $h$ is empty, but the region $h$ belongs to is full and no acceptable resident to $h$ is assigned to a hospital in the same region. %
 \end{description}
Consider the hospitals in $T_{j}$-gadget.
       For each $h \in \set{g'_{j,2},g'_{j,4}}$, $h$ is empty, but the region $h$ belongs to is full and no acceptable resident to $h$ is assigned to a hospital in the same region.
       Hospital $t'_{j}$ is assigned the first-choice resident $z'_{j}$.
 Thus, no hospital in $I'$ can be a part of an SBP{}{},
 and hence
 $M$ is a strongly stable matching{} of $I'$.

 Conversely, suppose that $I'$ admits a strongly stable matching{} $M$.
 Let $A$ be an assignment of $I$
 constructed as follows:
 if $(c'_{j_{i,3},\ell_{i,3}},x'_{i,3}) \in M$, then set $x_{i} = 1$;
 otherwise, set $x_{i} = 0$.
 We show that $A$ is a satisfying assignment of $I$.

 We first see some properties of $M$.
 For a $T_{j}$-gadget,
 $(z'_{j},y'_{j}) \not\in M$ holds.
 For, if $(z'_{j},y'_{j}) \in M$,
 then $g'_{j,1}$, $g'_{j,2}$, $g'_{j,3}$, and $g'_{j,4}$
 form the same structure as $G_2$ in \cref{L14}.
 Hence $M$ has an SBP{}{}, a contradiction.
 For a $C_{j}$-gadget ($j \in [1,m_2]$),
 we show that
 (a) either $(c'_{j,1},x'_{i_{j,1},k_{j,1}}) \not\in M$ or $(c'_{j,2},x'_{i_{j,2},k_{j,2}}) \not\in M$ holds.
 Suppose not.
 Then, both $(c'_{j,1},x'_{i_{j,1},k_{j,1}}) \in M$ and $(c'_{j,2},x'_{i_{j,2},k_{j,2}}) \in M$ hold.
 Recall that $(z'_{j},y'_{j}) \not\in M$.
 Since both $a'_{j,1}$ and $y'_{j}$ are empty
 and $y'_{j}$ is the first-choice hospital of $z'_{j}$,
 $(z'_{j}, y'_{j})$ is an SBP{}{}, a contradiction.
 For a $C_{j}$-gadget ($j \in [m_2+1,m]$),
 we show that
 (b) either $(c'_{j,1},x'_{i_{j,1},k_{j,1}}) \not\in M$, $(c'_{j,2},x'_{i_{j,2},k_{j,2}}) \not\in M$, or $(c'_{j,3},x'_{i_{j,3},k_{j,3}}) \not\in M$ holds.
 Suppose not.
 Then, all of $(c'_{j,1},x'_{i_{j,1},k_{j,1}}) \in M$, $(c'_{j,2},x'_{i_{j,2},k_{j,2}}) \in M$, and $(c'_{j,3},x'_{i_{j,3},k_{j,3}}) \in M$ hold.
 Since $a'_{j,1}$ is empty and $a'_{j,2}$ is the first-choice hospital of $d'_{j}$,
 $a'_{j,2}$ must be assigned $d'_{j}$ in $M$ to avoid forming an SBP{}{}.
 Then, both $(d'_{j},a'_{j,2}) \in M$ and $(c'_{j,3},x'_{i_{j,3},k_{j,3}}) \in M$ hold.
 Recall that $(z'_{j},y'_{j}) \not\in M$.
 Since $a'_{j,3}$ is empty and $y'_{j}$ is the first-choice hospital of $z'_{j}$,
 $(z'_{j}, y'_{j})$ is an SBP{}{}, a contradiction.
 For an $x_{i}$-gadget,
 we show that
 (c) either (i) $(c'_{j_{i,3},\ell_{i,3}},x'_{i,3}) \in M$ or (ii) $(c'_{j_{i,1},\ell_{i,1}},x'_{i,1}) \in M$ and $(c'_{j_{i,2},\ell_{i,2}},x'_{i,2}) \in M$ holds.
 Suppose that
 $(c'_{j_{i,3},\ell_{i,3}},x'_{i,3}) \not\in M$.
 Then since
 $x'_{i,3}$ is the first-choice of $c'_{j_{i,3},\ell_{i,3}}$,
 both $(e'_{i,1}, b'_{i,3}) \in M $ and $(e'_{i,2}, b'_{i,4}) \in M$ hold
 (otherwise, $( c'_{j_{i,3},\ell_{i,3}}, x'_{i,3} )$ is an SBP{}{} for $M$, a contradiction).
 Then,
 $b'_{i,1}$ and $b'_{i,2}$ are empty.
 If $(c'_{j_{i,1},\ell_{i,1}},x'_{i,1}) \not\in M$,
 then $( c'_{j_{i,1},\ell_{i,1}}, x'_{i,1} )$ is an SBP{}{} for $M$,
 since $c'_{j_{i,1},\ell_{i,1}}$ is the first-choice of $c'_{j_{i,1},\ell_{i,1}}$. %
 Similarly, if $(c'_{j_{i,2},\ell_{i,2}},x'_{i,2}) \not\in M$,
 then $( c'_{j_{i,2},\ell_{i,2}}, x'_{i,2} )$ is an SBP{}{} for $M$.
 Therefore, we have $(c'_{j_{i,1},\ell_{i,1}},x'_{i,1}) \in M$ and $(c'_{j_{i,2},\ell_{i,2}},x'_{i,2}) \in M$.

{

 By the construction of $A$, we have
 (d1) if $x_{i} = 1$, then $(c'_{j_{i,3},\ell_{i,3}},x'_{i,3}) \in M$; and
 (d2) if $x_{i} = 0$, then $(c'_{j_{i,3},\ell_{i,3}},x'_{i,3}) \not\in M$.
 By (d2) and (c), we have
 (d3) if $x_{i} = 0$, then $(c'_{j_{i,1},\ell_{i,1}},x'_{i,1}) \in M$ and $(c'_{j_{i,2},\ell_{i,2}},x'_{i,2}) \in M$.

 Now suppose that $A$ is not a satisfying assignment.
 Let $C_{j}$ be an unsatisfied clause
 and consider its $\ell$th literal ($\ell = 1, 2, 3$).
 Recall that this literal is of variable $x_{i_{j,\ell}}$.
 Consider three cases depending on whether $k_{j,\ell}$ is 1, 2, or 3.
 If $k_{j,\ell} = 1$,
 then the $\ell$th literal of $C_{j}$ is the first positive occurrence of $x_{i_{j,\ell}}$.
 Since $C_{j}$ is unsatisfied, $x_{i_{j,\ell}} = 0$ holds and by (d3),
 $(c'_{j,\ell},x'_{i_{j,\ell},1}) \in M$.
 If $k_{j,\ell} = 2$,
 then the $\ell$th literal of $C_{j}$ is the second positive occurrence of $x_{i_{j,\ell}}$,
 and by a similar argument, we have that
 $(c'_{j,\ell},x'_{i_{j,\ell},2}) \in M$.
 If $k_{j,\ell} = 3$,
 then the $\ell$th literal of $C_{j}$ is the (unique) negative occurrence of $x_{i_{j,\ell}}$.
 Since $C_{j}$ is unsatisfied, $x_{i_{j,\ell}} = 1$ holds and by (d1),
 $(c'_{j,\ell},x'_{i_{j,\ell},3}) \in M$.
 Thus, in any case, we have $(c'_{j,\ell},x'_{i_{j,\ell},k_{j,\ell}}) \in M$.
 When $C_{j}$ is a $2$-clause,
 we have both $(c'_{j,1},x'_{i_{j,1},k_{j,1}}) \in M$ and $(c'_{j,2},x'_{i_{j,2},k_{j,2}}) \in M$,
 but this contradicts (a).
 When $C_{j}$ is a $3$-clause,
 we have that $(c'_{j,1},x'_{i_{j,1},k_{j,1}}) \in M$, $(c'_{j,2},x'_{i_{j,2},k_{j,2}}) \in M$ and $(c'_{j,3},x'_{i_{j,3},k_{j,3}}) \in M$,
 but this contradicts (b).
 Thus, $A$ is a satisfying assignment of $I$, and the proof is completed.
}
\end{proof}

  \subsection{NP-completeness of \textsc{Strong-$(2,3,2)$-HRRCDR} }
  \label{L38}
\begin{theorem}
 \label{L9}

 \textsc{Strong-$(2,3,2)$-HRRCDR} is NP-complete.
\end{theorem}
\begin{proof}
{

\newcommand{\hq}[1]{q\hifempty{#1}{}{(#1)}} %

\newcommand{\hc}[1]{c\hifempty{#1}{}{(#1)}} %

\newcommand{\hpreflist}[2]{
\begin{tabular}{rccccccccccccccccccccccc}
 #1
\end{tabular}
\hspace{1cm}
\begin{tabular}{rccccccccccccccccccccccc}
 #2
\end{tabular}
}

 Membership in NP is obvious.
 We show a reduction from an NP-complete problem \textsc{PPN-3-SAT}{}.
 Let $I$ be an instance of \textsc{PPN-3-SAT}{} having
 $n$ variables $x_{i}$ ($i \in [1,n]$) and
 $m$ clauses $C_{j}$ ($j \in [1,m]$).
 For $k = 2, 3$, we call a clause containing $k$ literals a
 {\em $k$-clause}.
 Suppose that there are
 $m_2$ $2$-clauses
 and
 $m_3$ $3$-clauses
 (thus $m_2 + m_3 = m$),
 and assume without loss of generality that
 $C_{j}$ ($j \in [1,m_2]$) are $2$-clauses
 and
 $C_{j}$ ($j \in [m_2+1,m]$) are $3$-clauses.
 For each variable $x_{i}$, we construct a
 {\em variable gadget{}}.
 It consists of
 two residents $e'_{i,1}$ and $e'_{i,2}$,
 five hospitals
 $b'_{i,1}$, $x'_{i,1}$,
 $b'_{i,2}$, $x'_{i,2}$, and
 $x'_{i,3}$,
 and
 two regions
 $\set{b'_{i,1}, x'_{i,1}}$ and
 $\set{b'_{i,2}, x'_{i,2}}$.
 A variable gadget{} corresponding to $x_{i}$ is called an {\em $x_{i}$-gadget}.
 For each clause $C_{j}$,
 we construct a {\em clause gadget{}}.
 If $C_{j}$ is a $2$-clause, we create
 two residents $c'_{j,1}$ and $c'_{j,2}$,
 two hospitals $a'_{j,1}$ and $y'_{j}$,
 and
 a region $\set{a'_{j,1}, y'_{j}}$.
 If $C_{j}$ is a $3$-clause, we create
 four residents $c'_{j,1}$, $c'_{j,2}$, $d'_{j}$, and $c'_{j,3}$,
 four hospitals $a'_{j,1}$, $a'_{j,2}$, $a'_{j,3}$, and $y'_{j}$,
 and
 two regions $\set{a'_{j,1}, a'_{j,2}}$ and $\set{a'_{j,3}, y'_{j}}$.
 A clause gadget{} corresponding to $C_{j}$ is called a {\em $C_{j}$-gadget}.
 For each clause $C_{j}$,
 we also construct a {\em terminal gadget{}}.
 We create
 three residents $g'_{j,1}$, $g'_{j,3}$, and $z'_{j}$,
 two hospitals $g'_{j,2}$ and $g'_{j,4}$,
 and
 a region
 $\set{g'_{j,2}, g'_{j,4}}$.
 A terminal gadget{} corresponding to $C_{j}$ is called a {\em $T_{j}$-gadget}.
 Thus, there are
 $2n + 2m_2 + 4m_3 + 3m$
 residents,
 $5n + 2m_2 + 4m_3 + 2m$
 hospitals, and
 $2n +  m_2 + 2m_3 +  m$
 regions
 in the created \textsc{HRRC}{} instance, denoted $I'$.

 \begin{figure}[tbp]
  \centering
  \hpreflist{
  $e'_{i,1}$: & $b'_{i,1}$ & $x'_{i,3}$ \\
  \\
  \\
  $e'_{i,2}$: & $b'_{i,2}$ & $x'_{i,3}$ \\
  \\
  \\
  \\
  }{
  $b'_{i,1}[1]$: & $e'_{i,1}$ \\
  $x'_{i,1}[1]$: & $\underline{ c'_{j_{i,1},\ell_{i,1}} }$ \\ %
  \\
  $b'_{i,2}[1]$: & $e'_{i,2}$ \\
  $x'_{i,2}[1]$: & $\underline{ c'_{j_{i,2},\ell_{i,2}} }$ \\ %
  \\
  $x'_{i,3}[2]$: & $e'_{i,1}$ & $e'_{i,2}$ & $\underline{ c'_{j_{i,3},\ell_{i,3}} }$ \\ %
  } \\
  \begin{tabular}{c}
  $\hc{ \set{b'_{i,1}, x'_{i,1}} } := 1$,
  $\hc{ \set{b'_{i,2}, x'_{i,2}} } := 1$
  \end{tabular}
  \caption{Preference lists, capacities, and regional caps of $x_{i}$-gadget. }
  \label{L69}
 \end{figure}
 Let us construct preference lists of variable gadgets.
 Suppose that $x_{i}$'s $k$th positive occurrence ($k = 1, 2$) is in the
 $j_{i,k}$th clause $C_{j_{i,k}}$ as
 the $\ell_{i,k}$th literal ($1 \le \ell_{i,k} \le 3$).
 Similarly, suppose that $x_{i}$'s negative occurrence is in the
 $j_{i,3}$th clause $C_{j_{i,3}}$ as the $\ell_{i,3}$th literal ($1 \le \ell_{i,3} \le 3$).
 Then, preference lists, capacities, and regional caps of the $x_{i}$-gadget are
 constructed as shown in \cref{L69}.
 The clause gadget{}s are exactly the same as that
in the proof of \cref{L8}.
 Define $k_{j,\ell}$ as
 \[
 k_{j,\ell} := \pcase{ll}{
 1 & \text{if this is the first positive occurrence of $x_{i_{j,\ell}}$,} \\
 2 & \text{if this is the second positive occurrence of $x_{i_{j,\ell}}$,} \\
 3 & \text{if this is the negative occurrence of $x_{i_{j,\ell}}$.}
 }
 \]
The preference lists, capacities, and regional caps of the $C_{j}$-gadget
for a $2$-clause (respectively, a $3$-clause) are given in \cref{L66} (respectively, \cref{L67}).
 \begin{figure}[tbp]
  \centering
  \hpreflist{
  $g'_{j,1}$: & $g'_{j,2}$ & $g'_{j,4}$ \\
  $g'_{j,3}$: & $g'_{j,4}$ & $g'_{j,2}$ \\
  $z'_{j}$: & $\underline{ y'_{j} }$ & $g'_{j,2}$ \\
  }{
  $g'_{j,2}[1]$: & $z'_{j}$ & $g'_{j,3}$ & $g'_{j,1}$ \\
  $g'_{j,4}[1]$: &          & $g'_{j,1}$ & $g'_{j,3}$ \\
  \\
  } \\
  \begin{tabular}{c}
  $\hc{ \set{g'_{j,2}, g'_{j,4}} } := 1$
  \end{tabular}
  \caption{Preference lists, capacities, and regional caps of $T_{j}$-gadget. }
  \label{L70}
 \end{figure}
 Finally, we construct preference lists of terminal gadget{}s.
 The preference lists, capacities, and regional caps of
 the $T_{j}$-gadget are shown in \cref{L70}.
 Now the reduction is completed.
 It is not hard to see that the reduction
 can be performed in polynomial time and
 $I'$ is an instance of $(2,3,2)$-\textsc{HRRCDR}{}.

 Now we start a correctness proof.
 First, suppose that $I$ is satisfiable and
 let $A$ be a satisfying assignment.
 We construct a strongly stable matching{} $M$ of
 $I'$ from $A$ as follows.
 For an $x_{i}$-gadget,
 we consider the following two cases.
 \begin{itemize}
  \item If $x_{i} = 0$ by $A$, then add $(e'_{i,1}, x'_{i,3})$ and $(e'_{i,2}, x'_{i,3})$ to $M$.
  \item If $x_{i} = 1$ by $A$, then add $(e'_{i,1}, b'_{i,1})$ and $(e'_{i,2}, b'_{i,2})$ to $M$.
 \end{itemize}
 For a $C_{j}$-gadget ($j \in [1,m_2]$),
 we use the same construction as in the proof of \cref{L8}.
 For a $C_{j}$-gadget ($j \in [m_2+1,m]$),
 we again use the same construction as in the proof of \cref{L8}.
 For a $T_{j}$-gadget, we add $(z'_{j}, g'_{j,2})$ to $M$.
 It is easy to see that $M$ satisfies all the capacity constraints
 of hospitals and regions, and hence is a feasible{}{} matching.

 Next, we prove that $M$ is strongly stable in $I'$
 by showing that no hospital in $I'$ can be a part of an SBP{}{}.
For hospitals in $x_{i}$-gadget,
 we consider two cases.
 \begin{description}
  \item[{\boldmath $x_{i} = 0$.}]
	     For each $h \in \set{b'_{i,1},b'_{i,2}}$, $h$ is empty, but the region $h$ belongs to is full and no acceptable resident to $h$ is assigned to a hospital in the same region.
	     For each $h \in \set{x'_{i,1},x'_{i,2}}$, $h$ is assigned the first-choice resident.
	     Hospital $x'_{i,3}$,
	     whose capacity is two,
	     is full and assigned its
	     first-choice resident $e'_{i,1}$ and
	     second-choice resident $e'_{i,2}$
  \item[{\boldmath $x_{i} = 1$}.]
	     For each $h \in \set{b'_{i,1},b'_{i,2}}$, $h$ is assigned the first-choice resident.
	     For each $h \in \set{x'_{i,1},x'_{i,2}}$, $h$ is empty, but the region $h$ belongs to is full and no acceptable resident to $h$ is assigned to a hospital in the same region.
	     Hospital $x'_{i,3}$
	     is undersubscribed,
	     but
	     $e'_{i,1}$ is assigned to
	     a better hospital than $x'_{i,3}$
	     and
	     $e'_{i,2}$ is also assigned to
	     a better hospital than $x'_{i,3}$.
 \end{description}
No hospital in $C_{j}$-gadget can be a part of an SBP{}{}
 as shown in \cref{L8}.
 Consider the hospitals in $T_{j}$-gadget.
 Hospital $g'_{j,2}$ is assigned the first-choice resident $z'_{j}$.
 Hospital $g'_{j,4}$ is empty, but the region $g'_{j,4}$ belongs to is full and no acceptable resident to $g'_{j,4}$ is assigned to a hospital in the same region.

 Conversely, suppose that $I'$ admits a strongly stable matching{} $M$.
 Let $A$ be an assignment of $I$
 constructed as follows:
 if $(c'_{j_{i,3},\ell_{i,3}},x'_{i,3}) \in M$, then set $x_{i} = 1$;
 otherwise, set $x_{i} = 0$.
 We show that $A$ is a satisfying assignment of $I$.

 We first see some properties of $M$.
 For a $T_{j}$-gadget,
 $(z'_{j},y'_{j}) \not\in M$ holds.
 For, if $(z'_{j},y'_{j}) \in M$,
 then $g'_{j,1}$, $g'_{j,2}$, $g'_{j,3}$, and $g'_{j,4}$
 form the same structure as $G_2$ in \cref{L14}.
 Hence $M$ has an SBP{}{}, a contradiction.
 For a $C_{j}$-gadget ($j \in [1,m_2]$),
 as shown in \cref{L8},
 we have that
 (a) either $(c'_{j,1},x'_{i_{j,1},k_{j,1}}) \not\in M$ or $(c'_{j,2},x'_{i_{j,2},k_{j,2}}) \not\in M$ holds.
 For a $C_{j}$-gadget ($j \in [m_2+1,m]$),
 as shown in \cref{L8},
 we have that
 (b) either $(c'_{j,1},x'_{i_{j,1},k_{j,1}}) \not\in M$, $(c'_{j,2},x'_{i_{j,2},k_{j,2}}) \not\in M$, or $(c'_{j,3},x'_{i_{j,3},k_{j,3}}) \not\in M$ holds.
 For an $x_{i}$-gadget,
 we show that
 (c) either (i) $(c'_{j_{i,3},\ell_{i,3}},x'_{i,3}) \in M$ or (ii) $(c'_{j_{i,1},\ell_{i,1}},x'_{i,1}) \in M$ and $(c'_{j_{i,2},\ell_{i,2}},x'_{i,2}) \in M$ holds.
 Suppose that
 $(c'_{j_{i,3},\ell_{i,3}},x'_{i,3}) \not\in M$.
 Then since
 $x'_{i,3}$ is the first-choice hospital of $c'_{j_{i,3},\ell_{i,3}}$.
 $x'_{i,3}$ must be assigned
 both $e'_{i,1}$ and $e'_{i,2}$ to avoid forming an SBP{}{}.
 Since $e'_{i,1}$ prefers $b'_{i,1}$ to $x'_{i,3}$
 and $e'_{i,1}$ is the first-choice resident of $b'_{i,1}$,
 the region $\set{b'_{i,1},x'_{i,1}}$ must be full{}
 to avoid $(e'_{i,1},b'_{i,1})$ being an SBP{}{}.
 To make $\set{b'_{i,1},x'_{i,1}}$ full{},
 $x'_{i,1}$ must be assigned to $c'_{j_{i,1},\ell_{i,1}}$.
 By the same argument,
 $x'_{i,2}$ must also be assigned to $c'_{j_{i,2},\ell_{i,2}}$.
 Therefore, $(c'_{j_{i,1},\ell_{i,1}},x'_{i,1}) \in M$ and $(c'_{j_{i,2},\ell_{i,2}},x'_{i,2}) \in M$ hold.

{

 By the construction of $A$, we have
 (d1) if $x_{i} = 1$, then $(c'_{j_{i,3},\ell_{i,3}},x'_{i,3}) \in M$; and
 (d2) if $x_{i} = 0$, then $(c'_{j_{i,3},\ell_{i,3}},x'_{i,3}) \not\in M$.
 By (d2) and (c), we have
 (d3) if $x_{i} = 0$, then $(c'_{j_{i,1},\ell_{i,1}},x'_{i,1}) \in M$ and $(c'_{j_{i,2},\ell_{i,2}},x'_{i,2}) \in M$.

 Now suppose that $A$ is not a satisfying assignment.
 Let $C_{j}$ be an unsatisfied clause
 and consider its $\ell$th literal ($\ell = 1, 2, 3$).
 Recall that this literal is of variable $x_{i_{j,\ell}}$.
 Consider three cases depending on whether $k_{j,\ell}$ is 1, 2, or 3.
 If $k_{j,\ell} = 1$,
 then the $\ell$th literal of $C_{j}$ is the first positive occurrence of $x_{i_{j,\ell}}$.
 Since $C_{j}$ is unsatisfied, $x_{i_{j,\ell}} = 0$ holds and by (d3),
 $(c'_{j,\ell},x'_{i_{j,\ell},1}) \in M$.
 If $k_{j,\ell} = 2$,
 then the $\ell$th literal of $C_{j}$ is the second positive occurrence of $x_{i_{j,\ell}}$,
 and by a similar argument, we have that
 $(c'_{j,\ell},x'_{i_{j,\ell},2}) \in M$.
 If $k_{j,\ell} = 3$,
 then the $\ell$th literal of $C_{j}$ is the (unique) negative occurrence of $x_{i_{j,\ell}}$.
 Since $C_{j}$ is unsatisfied, $x_{i_{j,\ell}} = 1$ holds and by (d1),
 $(c'_{j,\ell},x'_{i_{j,\ell},3}) \in M$.
 Thus, in any case, we have $(c'_{j,\ell},x'_{i_{j,\ell},k_{j,\ell}}) \in M$.
 When $C_{j}$ is a $2$-clause,
 we have both $(c'_{j,1},x'_{i_{j,1},k_{j,1}}) \in M$ and $(c'_{j,2},x'_{i_{j,2},k_{j,2}}) \in M$,
 but this contradicts (a).
 When $C_{j}$ is a $3$-clause,
 we have that $(c'_{j,1},x'_{i_{j,1},k_{j,1}}) \in M$, $(c'_{j,2},x'_{i_{j,2},k_{j,2}}) \in M$ and $(c'_{j,3},x'_{i_{j,3},k_{j,3}}) \in M$,
 but this contradicts (b).
 Thus, $A$ is a satisfying assignment of $I$, and the proof is completed.
}
}
\end{proof}

  \subsection{NP-completeness of \textsc{Strong-$(3,2,2)$-HRRCDR} }
  \label{L37}
\begin{theorem}%
 \label{L10}

 \textsc{Strong-$(3,2,2)$-HRRCDR} is NP-complete.

\end{theorem}
\begin{proof}
{

\newcommand{\hq}[1]{q\hifempty{#1}{}{(#1)}} %

\newcommand{\hc}[1]{c\hifempty{#1}{}{(#1)}} %

\newcommand{\hpreflist}[2]{
\begin{tabular}{rccccccccccccccccccccccc}
 #1
\end{tabular}
\hspace{1cm}
\begin{tabular}{rccccccccccccccccccccccc}
 #2
\end{tabular}
}

 Membership in NP is obvious.
 We show a reduction from an NP-complete problem \textsc{PPN-3-SAT}{}.
 Let $I$ be an instance of \textsc{PPN-3-SAT}{} having
 $n$ variables $x_{i}$ ($i \in [1,n]$) and
 $m$ clauses $C_{j}$ ($j \in [1,m]$).
 For $k = 2, 3$, we call a clause containing $k$ literals a
 {\em $k$-clause}.
 Suppose that there are
 $m_2$ $2$-clauses
 and
 $m_3$ $3$-clauses
 (thus $m_2 + m_3 = m$),
 and assume without loss of generality that
 $C_{j}$ ($j \in [1,m_2]$) are $2$-clauses
 and
 $C_{j}$ ($j \in [m_2+1,m]$) are $3$-clauses.
 For each variable $x_{i}$, we construct a
 {\em variable gadget{}}.
 It consists of
 four residents $e'_{i,1}$, $e'_{i,2}$, $e'_{i,3}$, and $e'_{i,4}$,
 five hospitals $x'_{i,1}$, $x'_{i,2}$, $x'_{i,3}$, $b'_{i,1}$, and $b'_{i,2}$, and
 a region $\set{b'_{i,1}, b'_{i,2}}$.
 A variable gadget{} corresponding to $x_{i}$ is called an {\em $x_{i}$-gadget}.
 For each clause $C_{j}$,
 we construct a {\em clause gadget{}}.
 If $C_{j}$ is a $2$-clause, we create
 three residents $c'_{j,1}$, $c'_{j,2}$, and $u'_{j,1}$,
 and
 three hospitals $a'_{j,1}$, $a'_{j,2}$, and $y'_{j}$.
 If $C_{j}$ is a $3$-clause, we create
 six residents $c'_{j,1}$, $c'_{j,2}$, $u'_{j,1}$, $d'_{j}$, $c'_{j,3}$, and $u'_{j,2}$,
 and
 six hospitals $a'_{j,1}$, $a'_{j,2}$, $w'_{j}$, $a'_{j,4}$, $a'_{j,3}$, and $y'_{j}$.
 A clause gadget{} corresponding to $C_{j}$ is called a {\em $C_{j}$-gadget}.
 For each clause $C_{j}$,
 we also construct a {\em terminal gadget{}}.
 We create
 two residents $z'_{j}$ and $g'_{j,3}$,
 two hospitals $g'_{j,2}$ and $g'_{j,4}$, and
 a region $\set{g'_{j,2}, g'_{j,4}}$.
 A terminal gadget{} corresponding to $C_{j}$ is called a {\em $T_{j}$-gadget}.
 Thus, there are
 $4n + 3m_2 + 6m_3 + 2m$ residents, 
 $5n + 3m_2 + 6m_3 + 2m$ hospitals, and
 $n + m$ regions
 in the created \textsc{HRRC}{} instance, denoted $I'$.

 \begin{figure}[tbp]
  \centering
  \hpreflist{
  $e'_{i,1}$: & $b'_{i,1}$ & $x'_{i,1}$ \\
  \\
  $e'_{i,2}$: & $b'_{i,1}$ & $x'_{i,2}$ \\
  \\
  $e'_{i,3}$: & $b'_{i,2}$ & $x'_{i,3}$ \\
  \\
  \\
  $e'_{i,4}$: & $b'_{i,2}$ \\
  }{
  $x'_{i,1}[1]$: & $e'_{i,1}$ & $\underline{ c'_{j_{i,1},\ell_{i,1}} }$ \\
  \\
  $x'_{i,2}[1]$: & $e'_{i,2}$ & $\underline{ c'_{j_{i,2},\ell_{i,2}} }$ \\
  \\
  $x'_{i,3}[1]$: & $e'_{i,3}$ & $\underline{ c'_{j_{i,3},\ell_{i,3}} }$ \\
  \\
  $b'_{i,1}[2]$: & $e'_{i,2}$ & $e'_{i,1}$ \\
  $b'_{i,2}[2]$: & $e'_{i,4}$ & $e'_{i,3}$ \\
  } \\
  \begin{tabular}{c}
  $\hc{ \set{b'_{i,1}, b'_{i,2}} } := 2$
  \end{tabular}
  \caption{Preference lists, capacities, and regional caps of $x_{i}$-gadget. }
  \label{L71}
 \end{figure}
 Let us construct preference lists of variable gadgets.
 Suppose that $x_{i}$'s $k$th positive occurrence ($k = 1, 2$) is in the
 $j_{i,k}$th clause $C_{j_{i,k}}$ as
 the $\ell_{i,k}$th literal ($1 \le \ell_{i,k} \le 3$).
 Similarly, suppose that $x_{i}$'s negative occurrence is in the
 $j_{i,3}$th clause $C_{j_{i,3}}$ as the $\ell_{i,3}$th literal ($1 \le \ell_{i,3} \le 3$).
 Then, preference lists, capacities, and regional caps of the $x_{i}$-gadget are
 constructed as shown in \cref{L71}.
 \begin{figure}[tbp]
  \centering
  \hpreflist{
  $c'_{j,1}$: & $\underline{ x'_{i_{j,1},k_{j,1}} }$ & $a'_{j,1}$ \\
  $c'_{j,2}$: & $\underline{ x'_{i_{j,2},k_{j,2}} }$ & $a'_{j,2}$ \\
  $u'_{j,1}$: & $a'_{j,1}$ & $a'_{j,2}$ & $y'_{j}$ \\
  }{
  $a'_{j,1}[1]$: & $c'_{j,1}$ & $u'_{j,1}$ \\
  $a'_{j,2}[1]$: & $c'_{j,2}$ & $u'_{j,1}$ \\
  $y'_{j}[1]$: & $u'_{j,1}$ & $\underline{ z'_{j} }$ \\
  } \\
  \begin{tabular}{c}
   (No regional cap is defined.)
  \end{tabular}
  \caption{Preference lists, capacities, and regional caps of $C_{j}$-gadget ($j \in [1,m_2]$). }
  \label{L72}
 \end{figure}
 \begin{figure}[tbp]
  \centering
  \hpreflist{
  $c'_{j,1}$: & $\underline{ x'_{i_{j,1},k_{j,1}} }$ & $a'_{j,1}$ \\
  $c'_{j,2}$: & $\underline{ x'_{i_{j,2},k_{j,2}} }$ & $a'_{j,2}$ \\
  $u'_{j,1}$: & $a'_{j,1}$ & $a'_{j,2}$ & $w'_{j}$ \\
  \\
  $d'_{j}$:   & $w'_{j}$    & $a'_{j,4}$ \\
  $c'_{j,3}$: & $\underline{ x'_{i_{j,3},k_{j,3}} }$ & $a'_{j,3}$ \\
  $u'_{j,2}$: & $a'_{j,4}$ & $a'_{j,3}$ & $y'_{j}$ \\
  }{
  $a'_{j,1}[1]$: & $c'_{j,1}$ & $u'_{j,1}$ \\
  $a'_{j,2}[1]$: & $c'_{j,2}$ & $u'_{j,1}$ \\
  $w'_{j}[1]$: & $u'_{j,1}$ & $d'_{j}$ \\
  \\
  $a'_{j,4}[1]$: & $d'_{j}$   & $u'_{j,2}$ \\
  $a'_{j,3}[1]$: & $c'_{j,3}$ & $u'_{j,2}$ \\
  $y'_{j}[1]$: & $u'_{j,2}$ & $\underline{ z'_{j} }$ \\
  } \\
  \begin{tabular}{c}
   (No regional cap is defined.)
  \end{tabular}
  \caption{Preference lists, capacities, and regional caps of $C_{j}$-gadget ($j \in [m_2+1,m]$). }
  \label{L73}
 \end{figure}
 We then construct preference lists of clause gadget{}s.
 Consider a clause $C_{j}$, and suppose that its $\ell$th literal is of a
 variable $x_{i_{j,\ell}}$.
 Define $k_{j,\ell}$ as
 \[
 k_{j,\ell} := \pcase{ll}{
 1 & \text{if this is the first positive occurrence of $x_{i_{j,\ell}}$,} \\
 2 & \text{if this is the second positive occurrence of $x_{i_{j,\ell}}$,} \\
 3 & \text{if this is the negative occurrence of $x_{i_{j,\ell}}$.}
 }
 \]
 If $C_{j}$ is a $2$-clause (respectively, $3$-clause),
 then the preference lists, capacities, and regional caps of the $C_{j}$-gadget are shown in
 \cref{L72} (respectively, \cref{L73}).
 \begin{figure}[tbp]
  \centering
  \hpreflist{
  $z'_{j}$:   & $\underline{ y'_{j} }$ & $g'_{j,2}$ & $g'_{j,4}$ \\
  $g'_{j,3}$: &          & $g'_{j,4}$ & $g'_{j,2}$ \\
  }{
  $g'_{j,2}[1]$: & $g'_{j,3}$ & $z'_{j}$ \\
  $g'_{j,4}[1]$: & $z'_{j}$ & $g'_{j,3}$ \\
  } \\
  \begin{tabular}{c}
   $\hc{ \set{g'_{j,2}, g'_{j,4}} } := 1$
  \end{tabular}
  \caption{Preference lists, capacities, and regional caps of $T_{j}$-gadget. }
  \label{L74}
 \end{figure}
 Finally, we construct preference lists of terminal gadget{}s.
 The preference lists, capacities, and regional caps of
 the $T_{j}$-gadget are shown in \cref{L74}.
 Now the reduction is completed.
 It is not hard to see that the reduction
 can be performed in polynomial time and
 $I'$ is an instance of $(3,2,2)$-\textsc{HRRCDR}{}.

 It might be helpful to roughly explain intuition behind the construction of $I'$.
 Consider the $x_{i}$-gadget in \cref{L71}.
 Three hospitals $x'_{i,1}$, $x'_{i,2}$, and $x'_{i,3}$ respectively
 correspond to the first positive occurrence,
 the second positive occurrence, and the negative occurrence of $x_{i}$.
 There are two ``proper'' assignment for an $x_{i}$-gadget,
 $M_{i,0}=\set{ (e'_{i,1}, x'_{i,1}), (e'_{i,2}, x'_{i,2}), (e'_{i,3}, b'_{i,2}), (e'_{i,4}, b'_{i,2}) }$
 and
 $M_{i,1}=\set{ (e'_{i,1}, b'_{i,1}), (e'_{i,2}, b'_{i,1}), (e'_{i,3}, x'_{i,3}) }$.
 We associate an assignment $x_{i}=0 $ with $M_{i,0}$ and $x_{i}=1$ with $M_{i,1}$.
 If $M_{i,0}$ is chosen, the hospital $x'_{i,3}$,
 corresponding to $\overline{x_{i}}$,
 is ``available'' in the sense that
 $x'_{i,3}$ is currently unassigned
 and can be assigned to $c'_{j_{i,3},\ell_{i,3}}$ later.
 Similarly, if $M_{i,1}$ is chosen,
 the hospitals $x'_{i,1}$ and $x'_{i,2}$,
 corresponding to $x_{i}$,
 are available.
 Note that a matching other than these two are not beneficial.
 For example, adopting a matching
 $\set{ (e'_{i,1}, x'_{i,1}), (e'_{i,2}, b'_{i,1}), (e'_{i,3}, x'_{i,3}), (e'_{i,4}, b'_{i,2}) }$
 makes $x'_{i,1}$ and $x'_{i,3}$ assigned,
 so only $x'_{i,2}$ is available, which is strictly worse than choosing $M_{i,1}$.

 Next, consider a $C_{j}$-gadget for a $2$-clause $C_{j}$ and see \cref{L72}.
 The residents $c'_{j,1}$ and $c'_{j,2}$
 correspond to the first and the second literals of $C_{j}$, respectively.
 If the $\ell$th literal ($\ell=1, 2$) takes the value 0 by an assignment,
 then $c'_{j,\ell}$ must be assigned to the second-choice hospital $a'_{j,\ell}$
 to avoid forming an SBP{}{} because
 the first-choice hospital $x'_{i_{j,\ell},k_{j,\ell}}$ is ``unavailable'' as we saw above
 and $c'_{j,\ell}$ is the first-choice hospital of $a'_{j,\ell}$.
 Therefore, if $C_{j}$ is unsatisfied by an assignment,
 both hospitals $a'_{j,1}$ and $a'_{j,2}$ are assigned
 their first-choice residents,
 $u'_{j,1}$ cannot be assigned to neither $a'_{j,1}$ nor $a'_{j,2}$.
 This forces $u'_{j,1}$ to be assigned to $y'_{j}$,
 since $u'_{j,1}$ is the first-choice resident of $y'_{j}$.
 Conversely, if $C_{j}$ is satisfied,
 one of $c'_{j,1}$ and $c'_{j,2}$ can be assigned to
 the hospital corresponding to a literal taking the value 1.
 Then, $u'_{j,1}$ can be assigned to at least one of $a'_{j,1}$ and $a'_{j,2}$.
 The argument for a $3$-clause is a bit more complicated,
 but the same conclusion holds:
 $C_{j}$ is satisfied if and only if
 $y'_{j}$ is not assigned its first-choice hospital.
 Then $y'_{j}$ is available if and only if $C_{j}$ is satisfied.

 Note that the hospital $y'_{j}$ in the $C_{j}$-gadget
 is the first-choice hospital of the resident $z'_{j}$ in the $T_{j}$-gadget.
 Therefore, if $y'_{j}$ is unavailable,
 we cannot assign $z'_{j}$ to $y'_{j}$.
 In this case, the $T_{j}$-gadget reduces to the instance $G_{2}$ in \cref{L14},
 so at least one SBP{}{} yields from this gadget.
 On the other hand, if $y'_{j}$ is available,
 we can assign $z'_{j}$ to $y'_{j}$,
 so the $T_{j}$-gadget will not create an SBP{}{}.
 In this way, we can associate a satisfying assignment of $I$
 with a strongly stable matching{} of $I'$. 

 Now we start a formal correctness proof.
 First, suppose that $I$ is satisfiable and
 let $A$ be a satisfying assignment.
 We construct a strongly stable matching{} $M$ of
 $I'$ from $A$ as follows.
 We say that a $2$-clause $C$ is
 {\em $(v_1,v_2)$-satisfied} by $A$ ($v_i\in \set{0,1}$)
 if the first and the second literals of $C$ take the values $v_1$ and $v_2$,
 respectively, in $A$.
 A $3$-clause is defined to be $(v_1,v_2,v_3)$-satisfied analogously.
 For an $x_{i}$-gadget,
 we consider the following two cases.
 \begin{itemize}
  \item If $x_{i} = 0$ by $A$, then add
	$(e'_{i,1}, x'_{i,1})$,
	$(e'_{i,2}, x'_{i,2})$,
	$(e'_{i,3}, b'_{i,2})$, and
	$(e'_{i,4}, b'_{i,2})$
	to $M$.
  \item If $x_{i} = 1$ by $A$, then add
	$(e'_{i,1}, b'_{i,1})$,
	$(e'_{i,2}, b'_{i,1})$, and
	$(e'_{i,3}, x'_{i,3})$
	to $M$.
 \end{itemize}
 For a $C_{j}$-gadget ($j \in [1,m_2]$), we consider the following three cases.
 \begin{itemize}

  \item If $C_{j}$ is $(0,1)$-satisfied by $A$, then add $(c'_{j,1}, a'_{j,1})$, $(c'_{j,2},x'_{i_{j,2},k_{j,2}})$, and $(u'_{j,1}, a'_{j,2})$ to $M$.
  \item If $C_{j}$ is $(1,0)$-satisfied by $A$, then add $(c'_{j,1},x'_{i_{j,1},k_{j,1}})$, $(c'_{j,2}, a'_{j,2})$, and $(u'_{j,1}, a'_{j,1})$ to $M$.
  \item If $C_{j}$ is $(1,1)$-satisfied by $A$, then add $(c'_{j,1},x'_{i_{j,1},k_{j,1}})$, $(c'_{j,2},x'_{i_{j,2},k_{j,2}})$, and $(u'_{j,1},a'_{j,1})$ to $M$.
 \end{itemize}
 For a $C_{j}$-gadget ($j \in [m_2+1,m]$), we consider the following seven cases.
 \begin{itemize}

  \item If $C_{j}$ is $(0,0,1)$-satisfied by $A$, then add $(c'_{j,1}, a'_{j,1})$, $(c'_{j,2}, a'_{j,2})$, $(u'_{j,1}, w'_{j})$, $(d'_{j}, a'_{j,4})$, $(c'_{j,3}, x'_{i_{j,3},k_{j,3}})$, $(u'_{j,2}, a'_{j,3})$ to $M$.
  \item If $C_{j}$ is $(0,1,0)$-satisfied by $A$, then add $(c'_{j,1}, a'_{j,1})$, $(c'_{j,2}, x'_{i_{j,2},k_{j,2}})$, $(u'_{j,1}, a'_{j,2})$, $(d'_{j,}, w'_{j})$, $(c'_{j,3}, a'_{j,3})$,$(u'_{j,2}, a'_{j,4})$ to $M$.
  \item If $C_{j}$ is $(1,0,0)$-satisfied by $A$, then add $(c'_{j,1}, x'_{i_{j,1},k_{j,1}})$, $(c'_{j,2}, a'_{j,2})$, $(u'_{j,1}, a'_{j,1})$, $(d'_{j,}, w'_{j})$, $(c'_{j,3}, a'_{j,3})$, $(u'_{j,2}, a'_{j,4})$ to $M$.
  \item If $C_{j}$ is $(0,1,1)$-satisfied by $A$, then add $(c'_{j,1}, a'_{j,1})$, $(c'_{j,2}, x'_{i_{j,2},k_{j,2}})$, $(u'_{j,1}, a'_{j,2})$, $(d'_{j,}, w'_{j})$, $(c'_{j,3}, x'_{i_{j,3},k_{j,3}})$, $(u'_{j,2}, a'_{j,4})$ to $M$.
  \item If $C_{j}$ is $(1,0,1)$-satisfied by $A$, then add $(c'_{j,1}, x'_{i_{j,1},k_{j,1}})$, $(c'_{j,2}, a'_{j,2})$, $(u'_{j,1}, a'_{j,1})$, $(d'_{j,}, w'_{j})$, $(c'_{j,3}, x'_{i_{j,3},k_{j,3}})$, $(u'_{j,2}, a'_{j,4})$ to $M$.
  \item If $C_{j}$ is $(1,1,0)$-satisfied by $A$, then add $(c'_{j,1}, x'_{i_{j,1},k_{j,1}})$, $(c'_{j,2}, x'_{i_{j,2},k_{j,2}})$, $(u'_{j,1}, a'_{j,1})$, $(d'_{j,}, w'_{j})$, $(c'_{j,3}, a'_{j,3})$, $(u'_{j,2}, a'_{j,4})$ to $M$.
  \item If $C_{j}$ is $(1,1,1)$-satisfied by $A$, then add $(c'_{j,1}, x'_{i_{j,1},k_{j,1}})$, $(c'_{j,2}, x'_{i_{j,2},k_{j,2}})$, $(u'_{j,1}, a'_{j,1})$, $(d'_{j,}, w'_{j})$, $(c'_{j,3}, x'_{i_{j,3},k_{j,3}})$, $(u'_{j,2}, a'_{j,4})$ to $M$.
 \end{itemize}
 For a $T_{j}$-gadget, add $(z'_{j},y'_{j})$ and $(g'_{3},g'_{4})$ to $M$.
 It is easy to see that $M$ satisfies all the capacity constraints
 of hospitals and regions, and hence is a feasible{}{} matching.

 Next, we prove that $M$ is strongly stable in $I'$
 by showing that no hospital in $I'$ can be a part of an SBP{}{}.
For hospitals in $x_{i}$-gadget,
 we consider two cases.
 \begin{description}
  \item[{\boldmath $x_{i} = 0$.}]
	     For each $h \in \set{x'_{i,1},x'_{i,2}}$, $h$ is assigned the first-choice resident. %
	     Hospital $x'_{i,3}$ is assigned the second-choice resident $c'_{j_{i,3},\ell_{i,3}}$, and the first-choice resident $e'_{i,3}$ is assigned to a better hospital $b'_{i,2}$ than $x'_{i,3}$. %
	     Hospital $b'_{i,1}$ is empty, but the region $b'_{i,1}$ belongs to is full and no acceptable resident to $b'_{i,1}$ is assigned to a hospital in the same region. %
	     Hospital $b'_{i,2}$ is assigned all acceptable residents. %
  \item[{\boldmath $x_{i} = 1$}.]
	     For each $h \in \set{x'_{i,1},x'_{i,2}}$, $h$ is assigned the second-choice resident, and the first-choice resident is assigned to a better hospital than $h$. %
	     Hospital $x'_{i,3}$ is assigned the first-choice resident $e'_{i,3}$. %
	     Hospital $b'_{i,1}$ is assigned all acceptable residents. %
	     Hospital $b'_{i,2}$ is empty, but the region $b'_{i,2}$ belongs to is full and no acceptable resident to $b'_{i,2}$ is assigned to a hospital in the same region. %
 \end{description}
For hospitals in $C_{j}$-gadget ($j \in [1,m_2]$), we consider three cases.
 \begin{description}

\item[{\boldmath $C_{j}$ is $(0,1)$-satisfied.}]
				    Hospital $a'_{j,1}$ is assigned the first-choice resident $c'_{j,1}$. %
				    For each $h \in \set{a'_{j,2},y'_{j}}$, $h$ is assigned the second-choice resident, and the first-choice resident is assigned to a better hospital than $h$. %
\item[{\boldmath $C_{j}$ is $(1,0)$-satisfied.}]
				    For each $h \in \set{a'_{j,1},y'_{j}}$, $h$ is assigned the second-choice resident, and the first-choice resident is assigned to a better hospital than $h$. %
				    Hospital $a'_{j,2}$ is assigned the first-choice resident $c'_{j,2}$. %
\item[{\boldmath $C_{j}$ is $(1,1)$-satisfied.}]
				    For each $h \in \set{a'_{j,1},y'_{j}}$, $h$ is assigned the second-choice resident, and the first-choice resident is assigned to a better hospital than $h$. %
				    Hospital $a'_{j,2}$ is empty, but each acceptable resident to $a'_{j,2}$ is assigned to a better hospital than $a'_{j,2}$. %
 \end{description}
For hospitals in $C_{j}$-gadget ($j \in [m_2+1,m]$), we consider seven cases.

 \begin{description}

\item[{\boldmath $C_{j}$ is $(0,0,1)$-satisfied.}]
				    For each $h \in \set{a'_{j,1},a'_{j,2},w'_{j},a'_{j,4}}$, $h$ is assigned the first-choice resident. %
				    For each $h \in \set{a'_{j,3},y'_{j}}$, $h$ is assigned the second-choice resident, and the first-choice resident is assigned to a better hospital than $h$. %
\item[{\boldmath $C_{j}$ is $(0,1,0)$-satisfied.}]
				    For each $h \in \set{a'_{j,1},a'_{j,3}}$, $h$ is assigned the first-choice resident. %
				    For each $h \in \set{a'_{j,2},w'_{j},a'_{j,4},y'_{j}}$, $h$ is assigned the second-choice resident, and the first-choice resident is assigned to a better hospital than $h$. %
\item[{\boldmath $C_{j}$ is $(1,0,0)$-satisfied.}]
				    For each $h \in \set{a'_{j,1},w'_{j},a'_{j,4},y'_{j}}$, $h$ is assigned the second-choice resident, and the first-choice resident is assigned to a better hospital than $h$. %
				    For each $h \in \set{a'_{j,2},a'_{j,3}}$, $h$ is assigned the first-choice resident. %
\item[{\boldmath $C_{j}$ is $(0,1,1)$-satisfied.}]
				    Hospital $a'_{j,1}$ is assigned the first-choice resident $c'_{j,1}$. %
				    For each $h \in \set{a'_{j,2},w'_{j},a'_{j,4},y'_{j}}$, $h$ is assigned the second-choice resident, and the first-choice resident is assigned to a better hospital than $h$. %
				    Hospital $a'_{j,3}$ is empty, but each acceptable resident to $a'_{j,3}$ is assigned to a better hospital than $a'_{j,3}$. %
\item[{\boldmath $C_{j}$ is $(1,0,1)$-satisfied.}]
				    For each $h \in \set{a'_{j,1},w'_{j},a'_{j,4},y'_{j}}$, $h$ is assigned the second-choice resident, and the first-choice resident is assigned to a better hospital than $h$. %
				    Hospital $a'_{j,2}$ is assigned the first-choice resident $c'_{j,2}$. %
				    Hospital $a'_{j,3}$ is empty, but each acceptable resident to $a'_{j,3}$ is assigned to a better hospital than $a'_{j,3}$. %
\item[{\boldmath $C_{j}$ is $(1,1,0)$-satisfied.}]
				    For each $h \in \set{a'_{j,1},w'_{j},a'_{j,4},y'_{j}}$, $h$ is assigned the second-choice resident, and the first-choice resident is assigned to a better hospital than $h$. %
				    Hospital $a'_{j,2}$ is empty, but each acceptable resident to $a'_{j,2}$ is assigned to a better hospital than $a'_{j,2}$. %
				    Hospital $a'_{j,3}$ is assigned the first-choice resident $c'_{j,3}$. %
\item[{\boldmath $C_{j}$ is $(1,1,1)$-satisfied.}]
				    For each $h \in \set{a'_{j,1},w'_{j},a'_{j,4},y'_{j}}$, $h$ is assigned the second-choice resident, and the first-choice resident is assigned to a better hospital than $h$. %
				    For each $h \in \set{a'_{j,2},a'_{j,3}}$, $h$ is empty, but each acceptable resident to $h$ is assigned to a better hospital than $h$. %
 \end{description}
Consider the hospitals in $T_{j}$-gadget.
 Hospital $g'_{j,2}$ is empty, but each acceptable resident to $g'_{j,2}$ is assigned to a better hospital than $g'_{j,2}$. %
 Hospital $g'_{j,4}$ is assigned the second-choice resident $g'_{j,3}$, and the first-choice resident $z'_{j}$ is assigned to a better hospital $y'_{j}$ than $g'_{j,4}$. %

 Conversely, suppose that $I'$ admits a strongly stable matching{} $M$.
 Let $A$ be an assignment of $I$
 constructed as follows:
 if $(c'_{j_{i,3},\ell_{i,3}},x'_{i,3}) \not\in M$, then set $x_{i} = 1$;
 otherwise, set $x_{i} = 0$.
 We show that $A$ is a satisfying assignment of $I$.

 We first see some properties of $M$.
 For a $T_{j}$-gadget,
 $(z'_{j},y'_{j}) \in M$ holds.
 For, if $(z'_{j},y'_{j}) \not\in M$,
 then $z'_{j}$, $g'_{j,2}$, $g'_{j,3}$, and $g'_{j,4}$
 form the same structure as $G_2$ in \cref{L14}.
 Hence $M$ has an SBP{}{}, a contradiction.
 For a $C_{j}$-gadget ($j \in [1,m_2]$),
 we show that
 (a) either $(c'_{j,1},x'_{i_{j,1},k_{j,1}}) \in M$ or $(c'_{j,2},x'_{i_{j,2},k_{j,2}}) \in M$ holds,
 by using $(z'_{j},y'_{j}) \in M$.
 Suppose not.
 Then, both $(c'_{j,1},x'_{i_{j,1},k_{j,1}}) \not\in M$ and $(c'_{j,2},x'_{i_{j,2},k_{j,2}}) \not\in M$ hold.
 Since $c'_{j,1}$ is unassigned and $c'_{j,1}$ is the first-choice resident of $a'_{j,1}$,
 $c'_{j,1}$ must be assigned to $a'_{j,1}$ in $M$ to avoid forming an SBP{}{}.
 Similarly,
 since $c'_{j,2}$ is unassigned and $c'_{j,2}$ is the first-choice resident of $a'_{j,2}$,
 $c'_{j,2}$ must be assigned to $a'_{j,2}$. %
 Thus, $u'_{j,1}$ is not assigned to $a'_{j,1}$ nor $a'_{j,2}$.
 Recall that $(z'_{j},y'_{j}) \in M$.
 Since $u'_{j,1}$ is unassigned and $u'_{j,1}$ is the first-choice hospital of $y'_{j}$,
 $(u'_{j,1},y'_{j})$ is an SBP{}{}, a contradiction.
 For a $C_{j}$-gadget ($j \in [m_2+1,m]$),
 we show that
 (b) either $(c'_{j,1},x'_{i_{j,1},k_{j,1}}) \in M$, $(c'_{j,2},x'_{i_{j,2},k_{j,2}}) \in M$, or $(c'_{j,3},x'_{i_{j,3},k_{j,3}}) \in M$ holds.
 Suppose not.
 Then, all of $(c'_{j,1},x'_{i_{j,1},k_{j,1}}) \not\in M$, $(c'_{j,2},x'_{i_{j,2},k_{j,2}}) \not\in M$, and $(c'_{j,3},x'_{i_{j,3},k_{j,3}}) \not\in M$ hold.
 For contradiction, suppose $(d'_{j},w'_{j}) \in M$.
 Since
 $c'_{j,1}$, $c'_{j,2}$, $u'_{j,1}$, $a'_{j,1}$, $a'_{j,2}$, and $w'_{j}$
 have the same structure as
 $c'_{j',1}$, $c'_{j',2}$, $u'_{j',1}$, $a'_{j',1}$, $a'_{j',2}$, and $y'_{j'}$
 in $C_{j'}$-gadget ($j' \in [1,m_2]$)
 and
 we have
 $(c'_{j,1},x'_{i_{j,1},k_{j,1}}) \not\in M$,
 $(c'_{j,2},x'_{i_{j,2},k_{j,2}}) \not\in M$, and
 $(d'_{j},w'_{j}) \in M$,
 we can derive a contradiction by the same argument as above.
 Thus, we have $(d'_{j},w'_{j}) \not\in M$.
 Since
 $d'_{j}$, $c'_{j,3}$, $u'_{j,2}$, $a'_{j,4}$, $a'_{j,3}$, and $y'_{j}$
 have the same structure as
 $c'_{j',1}$, $c'_{j',2}$, $u'_{j',1}$, $a'_{j',1}$, $a'_{j',2}$, and $y'_{j'}$
 in $C_{j'}$-gadget ($j' \in [1,m_2]$),
 and
 we have 
 $(d'_{j},w'_{j}) \not\in M$,
 $(c'_{j,3},x'_{i_{j,3},k_{j,3}}) \not\in M$, and
 $(z'_{j},y'_{j}) \in M$,
 we can derive a contradiction again by the same argument as above.
 Thus, we have (b).
 For an $x_{i}$-gadget,
 we show that 
 (c) either (i) $(c'_{j_{i,3},\ell_{i,3}},x'_{i,3}) \not\in M$ or (ii) $(c'_{j_{i,1},\ell_{i,1}},x'_{i,1}) \not\in M$ and $(c'_{j_{i,2},\ell_{i,2}},x'_{i,2}) \not\in M$ holds.
 Suppose that
 $(c'_{j_{i,3},\ell_{i,3}},x'_{i,3}) \in M$.
 Then since
 $e'_{i,3}$ is the first-choice resident of $x'_{i,3}$,
 $e'_{i,3}$ must be assigned to $b'_{i,2}$ in $M$
 to avoid forming an SBP{}{}.
 Since 
 $b'_{i,2}$ prefers $e'_{i,4}$ to $e'_{i,3}$,
 $e'_{i,4}$ must also be assigned to $b'_{i,2}$ in $M$
 to avoid forming an SBP{}{}.
 Then since region $\set{ b'_{i,1}, b'_{i,2} }$ is full{},
 we have
 $(e'_{i,1}, b'_{i,1}) \not\in M$ and
 $(e'_{i,2}, b'_{i,2}) \not\in M$.
 Since $e'_{i,1}$ is unassigned and $e'_{i,1}$ is the first-choice resident of $x'_{i,1}$,
 $x'_{i,1}$ must be assigned $e'_{i,1}$ to avoid forming an SBP{}{}.
 Similarly, $x'_{i,2}$ must also be assigned $e'_{i,2}$.
 Thus, (ii) holds.

{

 By the construction of $A$, we have
 (d1) if $x_{i} = 1$, then $(c'_{j_{i,3},\ell_{i,3}},x'_{i,3}) \not\in M$; and
 (d2) if $x_{i} = 0$, then $(c'_{j_{i,3},\ell_{i,3}},x'_{i,3}) \in M$.
 By (d2) and (c), we have
 (d3) if $x_{i} = 0$, then $(c'_{j_{i,1},\ell_{i,1}},x'_{i,1}) \not\in M$ and $(c'_{j_{i,2},\ell_{i,2}},x'_{i,2}) \not\in M$.

 Now suppose that $A$ is not a satisfying assignment.
 Let $C_{j}$ be an unsatisfied clause
 and consider its $\ell$th literal ($\ell = 1, 2, 3$).
 Recall that this literal is of variable $x_{i_{j,\ell}}$.
 Consider three cases depending on whether $k_{j,\ell}$ is 1, 2, or 3.
 If $k_{j,\ell} = 1$,
 then the $\ell$th literal of $C_{j}$ is the first positive occurrence of $x_{i_{j,\ell}}$.
 Since $C_{j}$ is unsatisfied, $x_{i_{j,\ell}} = 0$ holds and by (d3),
 $(c'_{j,\ell},x'_{i_{j,\ell},1}) \not\in M$.
 If $k_{j,\ell} = 2$,
 then the $\ell$th literal of $C_{j}$ is the second positive occurrence of $x_{i_{j,\ell}}$,
 and by a similar argument, we have that
 $(c'_{j,\ell},x'_{i_{j,\ell},2}) \not\in M$.
 If $k_{j,\ell} = 3$,
 then the $\ell$th literal of $C_{j}$ is the (unique) negative occurrence of $x_{i_{j,\ell}}$.
 Since $C_{j}$ is unsatisfied, $x_{i_{j,\ell}} = 1$ holds and by (d1),
 $(c'_{j,\ell},x'_{i_{j,\ell},3}) \not\in M$.
 Thus, in any case, we have $(c'_{j,\ell},x'_{i_{j,\ell},k_{j,\ell}}) \not\in M$.
 When $C_{j}$ is a $2$-clause,
 we have both $(c'_{j,1},x'_{i_{j,1},k_{j,1}}) \not\in M$ and $(c'_{j,2},x'_{i_{j,2},k_{j,2}}) \not\in M$,
 but this contradicts (a).
 When $C_{j}$ is a $3$-clause,
 we have that $(c'_{j,1},x'_{i_{j,1},k_{j,1}}) \not\in M$, $(c'_{j,2},x'_{i_{j,2},k_{j,2}}) \not\in M$ and $(c'_{j,3},x'_{i_{j,3},k_{j,3}}) \not\in M$,
 but this contradicts (b).
 Thus, $A$ is a satisfying assignment of $I$, and the proof is completed.
}
}
\end{proof}

\bibliography{sm}

\end{document}